\documentclass[%
 reprint,
 amsmath,amssymb,amsthm,
 aps,
 prx
]{revtex4-1}

\usepackage{graphicx}
\usepackage{dcolumn,amsthm}
\usepackage{bm}

\usepackage{txfonts}
\usepackage{dsfont}
\usepackage[english]{babel}

\usepackage[colorlinks=true,linkcolor=blue,citecolor=blue,urlcolor=blue]{hyperref}

\newcommand{\<}{\langle}
\renewcommand{\>}{\rangle}
\newcommand{\abs}[1]{\left\lvert#1\right\rvert}
\newcommand{\norm}[1]{\lVert#1\rVert}

\providecommand{\tr}{{\rm tr}}
\renewcommand{\phi}{\varphi}
\newcommand{\bra}[1]{\left\langle #1\right\rvert}
\newcommand{\ket}[1]{\left\lvert #1\right\rangle}
\newcommand{\gf}[1]{\<#1\>_{{\bf g} \mid \vec x_{\bf f}}}
\newtheorem{theorem}{Theorem}

\begin{document}

\title{Conditional non-Gaussian quantum state preparation}

\author{Mattia Walschaers}
\email{mattia.walschaers@lkb.upmc.fr}
\author{Valentina Parigi}
\author{Nicolas Treps}
\affiliation{Laboratoire Kastler Brossel, Sorbonne Universit\'{e}, CNRS, ENS-PSL Research University, Coll\`{e}ge de France, 4 place Jussieu, F-75252 Paris, France}

\date{\today}

\begin{abstract}
We develop a general formalism, based on the Wigner function representation of continuous-variable quantum states, to describe the action of an arbitrary conditional operation on a multimode Gaussian state. We apply this formalism to several examples, thus showing its potential as an elegant analytical tool for simulating quantum optics experiments. Furthermore, we also use it to prove that EPR steering is a necessary requirement to remotely prepare a Wigner-negative state.
\end{abstract}

\maketitle


\section{\label{sec:Intro} Introduction}

In continuous-variable (CV) quantum physics, Gaussian states have long been a fruitful topic of research \cite{Schrodinger-Coherent,Kennard,PhysRev.131.2766, PhysRevLett.10.277,robinson_ground_1965,HUDSON1974249,PhysRevA.49.1567,RevModPhys.77.513,RevModPhys.84.621,doi:10.1142/S1230161214400010}. They appear naturally as the ground states of systems of many non-interacting particles in the form of thermal states \cite{verbeure_many-body_2011}, or as the coherent states that describe the light emitted by a laser \cite{PhysRev.131.2766}. Through nonlinear processes, it is possible to reduce the noise beyond the shot noise limit (at the price of increased noise in a complementary observable), and create squeezed states \cite{PhysRevA.13.2226,PhysRevLett.55.2409,PhysRevLett.57.2520,PhysRevLett.57.691,PhysRevLett.59.2555,PhysRevLett.117.110801}. For the purpose of metrology, such squeezed states are often enough to obtain a significant boost in performance \cite{PhysRevD.23.1693,Treps940,Ligo,SCHNABEL20171}.

On theoretical grounds, Gaussian states are relatively easy to handle \cite{RevModPhys.77.513,RevModPhys.84.621}. The quantum statistics of the continuous-variable observables (e.g., the quadratures in quantum optics) are described by Gaussian Wigner functions. All interesting quantum features can be deduced from the covariance matrix that characterises this Gaussian distribution on phase space. Hence, whenever the number of modes remains finite, the techniques of symplectic matrix analysis are sufficient to study Gaussian quantum states. This has generated an extensive understanding of the entanglement properties of Gaussian states \cite{PhysRevLett.84.2722,PhysRevLett.84.2726,PhysRevLett.86.3658,PhysRevA.71.055801,Adesso_2007,gerke_full_2015}, and recently it has also lead to the development of a measure for quantum steering (see \cite{RevModPhys.92.015001}) of Gaussian states with Gaussian measurements \cite{Wiseman:2007aa,Kogias:2015aa,Deng:2017aa,Cai-Steering}, which we refer to as Einstein-Podolsky-Rosen (EPR) steering.

Even though they have many advantages, Gaussian states are of limited use to quantum technologies beyond sensing. They have been shown to be easily simulated on classical devices \cite{PhysRevLett.88.097904}, and in particular Wigner negativity is known to be a necessary resource for reaching a quantum computation advantage \cite{mari_positive_2012}. However, it should be stressed that recent work has found large classes of Wigner negative states that can also be simulated easily \cite{garcalvarez2020efficient}. In other words, Wigner negativity is necessary but not sufficient to reach a quantum computation advantage \cite{rahimi-keshari_sufficient_2016}.

In the particular case of CV quantum computation, Gaussian states play an essential role in the measurement-based approach \cite{gu_quantum_2009}. In this paradigm, one establishes large Gaussian entangled states, known as cluster states, which form the backbone of the desired quantum routine \cite{PhysRevA.76.032321}. Several recent breakthroughs have led to the experimental realisation of such states \cite{Su:12,PhysRevLett.112.120505,cai-2017,Asavanant:2019aa,Larsen:2019aa}. Nevertheless, to execute quantum algorithms that cannot be simulated efficiently, one must induce Wigner negativity. In the spirit of measurement-based quantum computation, this feature is induced by measuring non-Gaussian observables, e.g., the number of photons, on a subset of modes \cite{PhysRevA.72.033822,PhysRevA.80.013806,PhysRevLett.113.100502,su2019generation}. Such a measurement then projects the remainder of the system into a non-Gaussian state. The exact properties of the resulting state depend strongly on the result of the measurement.

Such a conditional preparation of non-Gaussian quantum states is common procedure in quantum optics experiments \cite{lvovsky2020production}. Basic examples include the heralding of single-photon Fock states after parametric down-conversion \cite{PhysRevLett.56.58,PhysRevLett.87.050402,Bra_czyk_2010}, photon addition and subtraction \cite{Wenger:2004aa,Ourjoumtsev83,Zavatta660,parigi_probing_2007,averchenko_multimode_2016,Ra2019}, and known schemes to prepare more exotic states such as Schr\"odinger-cat \cite{dakna_generating_1997,Thekkadath2020engineering} or Gottesman-Kitaev-Preskill states \cite{Eaton_2019}. Remarkably, though, a practical framework to describe the effect of such conditional operations on arbitrary Gaussian states is still lacking. Notable exceptions where one does study arbitrary initial states usually rely on specific choices for the conditional measurement.

Here, in Section \ref{sec:Cond}, we introduce a practical framework to describe the resulting Wigner function for a quantum state that is conditionally prepared by measuring a subset of modes of a Gaussian multimode state. The techniques used in this work are largely based on classical multivariate probability theory and provide a conceptually new understanding of these conditioned states. In Section \ref{sec:steer}, we unveil the most striking consequence of this new framework: we can formally prove that EPR steering in the initial Gaussian state is a necessary requirement for the conditional preparation of Wigner-negative states, regardless of the measurement upon which we condition. This solidifies a previously conjectured general connection between EPR steering and Wigner-negativity. As shown in Section \ref{sec:examples}, our framework reproduce a range of known state-preparation schemes and can be used to treat more advanced scenarios, which could thus far not be addressed by other analytical methods. First, however, we review the phase space description of multimode CV systems in Section~\ref{sec:PreRec}.

\section{Phase space description of multimode continuous-variable systems}\label{sec:PreRec}

The CV approach studies quantum systems with an infinite-dimensional Hilbert space ${\cal H}$ based on observables, $\hat x$ and $\hat p$ that have a continuous spectrum and obey the canonical commutation relation $[\hat x, \hat p] = 2i$ (the factors two is chosen to normalise the vacuum noise to one). Common examples include the position and momentum operators in mechanical systems, or the amplitude and phase quadratures in quantum optics. In this work, we will use quantum optics terminology, but the results equally apply to any other system that is described by the algebra of canonical commutation relations (i.e., any bosonic system).

In a single-mode system, the quadrature observables $\hat x$ and $\hat p$ determine the optical phase space. The latter is a two-dimensional real space, where the axes denote the possible measurement outcomes for $\hat x$ and $\hat p$. It is common practice to represent a given state $\hat \rho$ by means of its measurement statistics for $\hat x$ and $\hat p$ on this optical phase space, as in statistical physics. However, because  $\hat x$ and $\hat p$ are complementary observables, they cannot be measured simultaneously, and thus, a priori, we cannot construct a joint probability distribution of phase space that reproduces the correct marginals to describe the measurement statistics the quadratures. Therefore, the phase space representation of quantum states are quasi-probability distributions. The quasi-probability distribution that reproduces the measurement statistics of the quadrature observables as its marginals, is known as the Wigner function \cite{PhysRev.40.749,PhysRev.177.1882,HILLERY1984121}
\begin{equation}
W(x,p) = \frac{1}{(2\pi)^2}\int_{\mathbb{R}^2} \tr[\hat \rho e^{i(\alpha_1 \hat x + \alpha_2 \hat p)}]e^{-i(\alpha_1 x + \alpha_2 p)} {\rm d}\alpha_1{\rm d}\alpha_2.\end{equation} For some quantum states, this function has the peculiar property of reaching negative values. This Wigner-negativity is a genuine hallmark of quantum physics, and it is understood to be crucial in reaching a quantum computational advantage.\\

Here, we will consider a multimode system comprising $m$ modes. Every mode comes with its own infinite-dimensional Hilbert space, associated to a two-dimensional phase space, and observables $\hat x_j$ and $\hat p_j$. The total optical phase space is, thus, a real space $\mathbb{R}^{2m}$ with a symplectic structure $\Omega = \bigoplus_m \omega$, where the two-dimensional matrix $\omega$ is given by
\begin{equation}
\omega= \begin{pmatrix}0 & -1 \\ 1 & 0\end{pmatrix}.
\end{equation}
Therefore, $\Omega$ has the properties $\Omega^2 = - \mathds{1}$ and $\Omega^T=-\Omega$.
Any normalised vector $\vec f \in \mathbb{R}^{2m}$ defines a single optical mode with an associated phase space ${\rm span}\{\vec f, \Omega \vec f\}$ (i.e., when $\vec f$ generates the phase space axis associated with the amplitude quadrature of this mode, $\Omega \vec f$ generates the axis for the associated phase quadrature). Henceforth, we will refer to the subsystem associated with the phase space ${\rm span}\{\vec f, \Omega \vec f\}$ as ``the mode $f$''. Every point $\vec \alpha \in \mathbb{R}^{2m}$ can also associated with a generalised quadrature observable
\begin{equation}
    \hat q(\vec \alpha) = \sum_{k=1}^{m} (\alpha_{2k-1} \hat{x}_{k} + \alpha_{2k} \hat{p}_{k}).
\end{equation}
These observables satisfy the general canonical commutation relation $[\hat q(\vec \alpha), \hat q(\vec \beta)] = -i \vec\alpha^T \Omega \vec\beta.$ Physically, such observable $\hat q(\vec \alpha)$ can be measured with a homodyne detector by selecting the mode that is determined by the direction of $\vec \alpha$, and multiplying the detector outcome by $\norm{\vec \alpha}$. In our theoretical treatment, such generalised quadratures are useful to define the quantum characteristic function of any multimode state $\hat \rho$
\begin{equation}
\chi_{\hat \rho}(\vec \alpha) =  \tr[\hat \rho \exp\{i \hat q (\vec \alpha)\}],
\end{equation}
for an arbitrary point $\vec \alpha$ in phase space. The multimode Wigner function of the state is then obtained as the Fourier transform of the characteristic function
\begin{equation}\label{eq:wigbasics}
W(\vec x) = \frac{1}{(2\pi)^{2m}}\int_{\mathbb{R}^{2m}} \chi_{\hat \rho}(\vec \alpha)e^{-i \vec \alpha^T \vec x} {\rm d}\vec \alpha,\end{equation}
where $\vec x \in \mathbb{R}^{2m}$ can, again, be any point in the multimode phase space, and the coordinates of $\vec x$ represent possible measurement outcomes for $\hat x_j$ and $\hat p_j$.

The Wigner function can be used to represent and characterise an arbitrary quantum state of the multimode system. In the same spirit, we can also define the phase space representation of an arbitrary observable ${\hat A}$ as
\begin{equation}\label{eq:WA-def}
W_{\hat A}(\vec x) = \frac{1}{(2\pi)^{2m}}\int_{\mathbb{R}^{2m}} \tr[\hat A \exp\{i \hat q (\vec \alpha)\}]e^{-i \vec \alpha^T \vec x} {\rm d}\vec \alpha,\end{equation}
such that we can fully describe the measurement statistics of an arbitrary quantum observable on phase space, by invoking the identity
\begin{equation}\label{eq:wigbasicsA}
\tr[\hat \rho \hat A] = (4\pi)^m \int_{\mathbb{R}^{2m}} W_{\hat A}(\vec x)W(\vec x)  {\rm d}\vec x\end{equation}
to evaluate expectation values. In practice, it is often challenging to obtain Wigner functions for arbitrary states or observables, but in some cases they can take convenient forms. 

A particular class of convenient states are Gaussian states, where the Wigner function $W(\vec x)$ is a Gaussian. As a consequence, the Wigner function is positive, and can thus be interpreted as a probability distribution. This Gaussian distribution is completely determined by a covariance matrix $V$, and mean-field $\vec \xi$, such that the Wigner function takes the form
\begin{equation}
W(\vec x) = \frac{e^{-\frac{1}{2} (\vec x-\vec \xi)^T V^{-1}(\vec x-\vec \xi)}}{(2\pi)^m \sqrt{\det V}}.
\end{equation}
This forms the basis of our preparation procedure for non-Gaussian states as we assume that our initial multimode system is prepared in such a Gaussian state.

To perform the conditional state-preparation, we divide the $m$-mode system in two subsets of orthogonal modes, ${\bf f} = \{f_1, \dots, f_l\}$ and ${\bf g} = \{g_1, \dots, g_{l'}\}$, with $l+l'=m$, and perform a measurement on the modes in ${\bf g}$. We can then describe the subsystems of modes ${\bf f}$ and ${\bf g}$ by phase spaces $\mathbb{R}^{2l}$ and $\mathbb{R}^{2l'}$, respectively. As such, the joint phase space can be mathematically decomposed as $\mathbb{R}^{2m} = \mathbb{R}^{2l} \oplus \mathbb{R}^{2l'}$. A general point $\vec{x}$ in the multimode phase space $\mathbb{R}^{2m}$ can thus be decomposed as $\vec{x} = \vec{x}_{\bf f} \oplus \vec{x}_{\bf g}$, where $\vec{x}_{\bf f}$ and $\vec{x}_{\bf g}$ describe the phase space coordinates associated with the sets of modes ${\bf f}$ and ${\bf g}$, respectively. In particular, $\vec{x}_{\bf f}$ can be expanded in a particular modes basis $f_1, \dots, f_l$ as $\vec{x}_{\bf f} = (x_{f_1}, p_{f_1}, \dots, x_{f_l}, p_{f_l})$, where the coordinates $x_{f_1}$ and $p_{f_1}$ are obtained as
\begin{align}
    &x_{f_j} = \vec x^T \vec f_j,\\
    &p_{f_j} = \vec x^T \Omega \vec f_j,
\end{align}
A completely analogous treatment is possible for the coordinates associated with the set of modes ${\bf g}$.


\section{Conditional operations in phase space}\label{sec:Cond}

\begin{figure}
\includegraphics[width=0.49\textwidth]{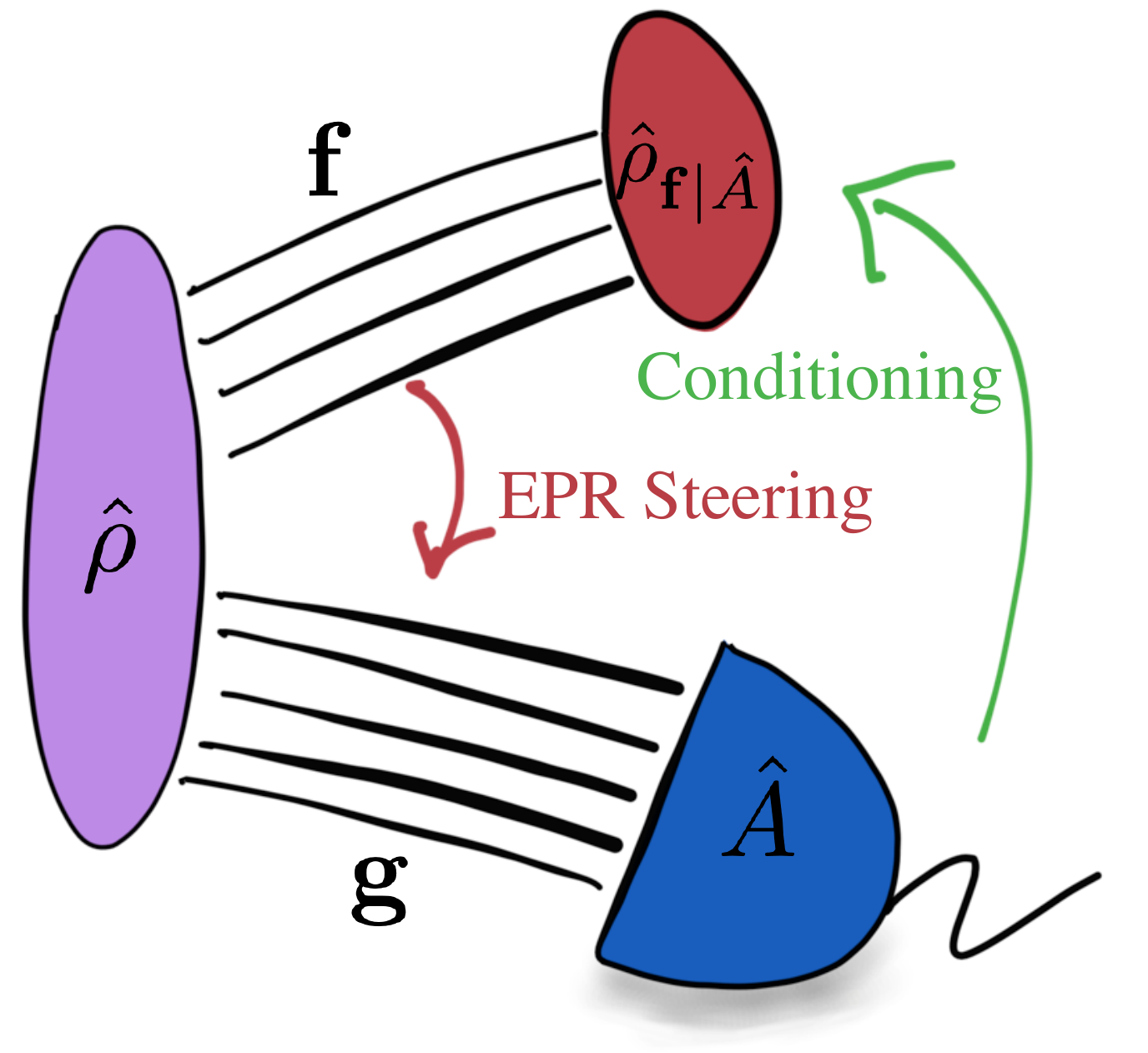}
\caption{Sketch of the conditional state preparation scenario: a multimode quantum state with density matrix $\hat \rho$ is separated over two subsets of modes, ${\bf f}$ and ${\bf g}$. A measurement is performed on the modes in ${\bf g}$, yielding a result associated with a POVM element $\hat A$. Conditioning on this measurement outcome ``projects'' the the subset of mode ${\bf f}$ into a state $\hat \rho_{{\bf f} \mid {\hat A}}$. The directional EPR steering, discussed in Section \ref{sec:steer}, is highlighted. \label{Fig:Sketch}}
\end{figure}

In quantum optics, we associate a Hilbert space (more precisely a Fock space) to each of these modes. The Hilbert space ${\cal H}$ of the entire system can then be structured as ${\cal H} = {\cal H}_{\bf f} \otimes {\cal H}_{\bf g}$, where $ {\cal H}_{\bf f}$ ( ${\cal H}_{\bf g}$) describes that quantum states of the set of orthogonal modes ${\bf f}$ (${\bf g}$). Formally, the state of our full $m$-mode system is then described by a density matrix $\hat \rho$ that acts on ${\cal H}$.

Within this manuscript, we perform a conditional operation in the set of modes ${\bf g}$, which we describe through a (not necessarily normalised) set of Kraus operators \cite{KRAUS1971311} $\hat X_j$ that act on ${\cal H}_{\bf g}$ \footnote{On the full multimode Hilbert space, the operator will take the form $\mathds{1} \otimes \hat X_j$, but for simplicity we will just denote it as $\hat X_j$.}:
\begin{equation}\label{eq:condFormal}
\hat \rho \mapsto \frac{\sum_j \hat X_j \hat \rho \hat X^{\dag}_j}{\tr [\sum_j \hat X^{\dag}_j \hat X_j \hat \rho]}.
\end{equation}
Such a conditional operation naturally arises as a post-measurement state, when $\hat X_j$ is a projector, or when $\hat A = \sum_j \hat X^{\dag}_j\hat X_j$ is a more general POVM element. The positive semi-definite operator $\hat A$ is useful to express the reduced state of the set of modes ${\bf f}$:
\begin{equation}\label{eq:condState}
\hat \rho_{{\bf f} \mid {\hat A}} = \frac{\tr_{\bf g} [\hat A \hat  \rho]}{\tr [\hat A \hat \rho]},
\end{equation}
where $\tr_{\bf g}$ denotes the partial trace of the Hilbert space ${\cal H}_{\bf g}$ associated with the set of mode ${\bf g}$.
Our general goal is to understand the properties of the state $\hat \rho_{{\bf f} \mid {\hat A}}$.

As we are interested in the Wigner function for the state of the subset of modes ${\bf f}$, we translate (\ref{eq:condState}) to its phase-space representation. We initialize the total system in a Gaussian state with Wigner function $W(\vec{x})$. Subsequently we also define the Wigner function $W_{\hat A}(\vec{x}_{\bf g})$ of the positive operator $\hat A$, which is a function that is defined according to (\ref{eq:WA-def}) on the phase space that describes the subset of modes $\bf g$. As such, we find that
\begin{equation}
\begin{split}\label{eq:wigcondinitial}
&W_{{\bf f} \mid \hat A }(\vec{x}_{\bf f}) = \frac{\int_{\mathbb{R}^{2l'}} W_{\hat A}(\vec{x}_{\bf g}) W(\vec{x}) {\rm d}\vec{x}_{\bf g}}{\int_{\mathbb{R}^{2m}} W_{\hat A}(\vec{x}_{\bf g}) W(\vec{x}) {\rm d}\vec{x}}.
\end{split}
\end{equation}
Because $\hat A$ is a positive semi-definite operator, the denominator is a positive constant.\\ 

As presented in (\ref{eq:wigcondinitial}), the Wigner function $W_{{\bf f} \mid \hat A }(\vec{x}_{\bf f})$ is impractical to use and its properties are not apparent. Hence, we now introduce some mathematical tools to obtain a more insightful expression for $W_{{\bf f} \mid \hat A }(\vec{x}_{\bf f})$. First, we use that, for Gaussian states, $W(\vec{x})$ is a probability distribution on phase space, such that we can define the conditional probability distribution through 
\begin{equation}
\begin{split}\label{eq:condProbDist}
W(\vec{x}_{\bf g} \mid \vec{x}_{\bf f}) = \frac{W(\vec{x})}{W_{\bf f}(\vec{x}_{\bf f})},
\end{split}
\end{equation}
where $W_{\bf f}(\vec{x}_{\bf f})$ is the reduced Gaussian state for the set of modes ${\bf f}$,
\begin{equation}
W_{\bf f}(\vec{x}_{\bf f}) = \int_{\mathbb{R}^{2l'}} W(\vec{x})  {\rm d}\vec{x}_{\bf g}.
\end{equation}
Because $W(\vec{x})$ is a Gaussian probability distribution, the conditional probability distribution $W(\vec{x}_{\bf g} \mid \vec{x}_{\bf f})$ is also a Gaussian distribution \cite{Muirhead:aa} with covariance matrix
\begin{equation}\label{eq:schur}
V_{{\bf g}\mid \vec x_{\bf f}} = V_{\bf g} - V_{{\bf gf}}V_{\bf f}^{-1}V_{\bf gf}^T,
\end{equation}
where $V_{\bf g}$ and $V_{\bf f}$ are the covariance matrices describing the subsets of modes ${\bf g}$ and ${\bf f}$ in the initial state, whereas $V_{\bf gf}$ describes all the initial Gaussian correlations between those subsets. Note that this covariance matrix is the same far all points $\vec{x}_{\bf f} \in \mathbb{R}^{2l}$, which is a particular property of Gaussian conditional probability distributions. 
Furthermore, the distribution $W(\vec{x}_{\bf g} \mid \vec{x}_{\bf f})$ also contains a displacement
\begin{equation}
\vec{\xi}_{{\bf g} \mid \vec x_{\bf f}} = \vec{\xi}_{\bf g} + V_{\bf gf} V_{\bf f}^{-1} (\vec{x}_{\bf f} - \vec{\xi}_{\bf f}),
\end{equation}
where $\vec{\xi}_{\bf g}$ and $\vec{\xi}_{\bf f}$ describe the displacements of the initial state in the sets of modes $\bf g$ and $\bf f$, respectively.

Generally, the phase space probability distribution $W(\vec{x}_{\bf g} \mid \vec{x}_{\bf f})$ is not a valid Wigner function of a well-defined quantum state, in the sense that it would violate the Heisenberg inequality. However, it does remain a well-defined probability distribution, i.e., it is normalised and positive. Thus, it still has interesting properties that we can exploit to formulate a general expression for $W_{{\bf f} \mid \hat A }(\vec{x}_{\bf f})$. Let us therefore define
\begin{equation}\label{eq:EA}
\<\hat A\>_{{\bf g} \mid \vec x_{\bf f}}= (4 \pi)^{l'} \int_{\mathbb{R}^{2l'}} W_{\hat A}(\vec{x}_{\bf g}) W(\vec{x}_{\bf g} \mid \vec{x}_{\bf f}) {\rm d}\vec{x}_{\bf g},
\end{equation}
which is the expectation value of the phase-space representation of $\hat A$ with respect to the probability distribution $W(\vec{x}_{\bf g} \mid \vec{x}_{\bf f})$. 
We can then use (\ref{eq:condProbDist}) and (\ref{eq:EA}) to recast (\ref{eq:wigcondinitial}) in the following form:
\begin{equation}\label{eq:Wf}
W_{{\bf f} \mid \hat A }(\vec{x}_{\bf f}) = \frac{\<\hat A\>_{{\bf g} \mid \vec x_{\bf f}}}{\<\hat A\>} W_{\bf f}(\vec{x}_{\bf f}),
\end{equation}
where we introduce the notation
\begin{equation}\label{eq:expA}
    \<\hat A\> =  \tr [\hat A \hat \rho] = (4\pi)^{l'}\int_{\mathbb{R}^{2m}} W_{\hat A}(\vec{x}_{\bf g}) W(\vec{x}) {\rm d}\vec{x}
\end{equation}
The major advantage of this formulation is that $\<\hat A\>_{{\bf g} \mid \vec x_{\bf f}}$ represents the average with respect to a Gaussian probability distribution, such that one can use several computational techniques that are well-know for Gaussian integrals. A notable property is the factorisation of higher moments in multivariate Gaussian distributions, such that $\<\hat A\>_{{\bf g} \mid \vec x_{\bf f}}$ can generally be expressed algebraically in terms of the components of $V_{{\bf g}\mid \vec x_{\bf f}}$ and $\vec{\xi}_{{\bf g} \mid \vec x_{\bf f}}$.\\

Finally, we remark that $\<\hat A\>_{{\bf g} \mid \vec x_{\bf f}} = \<\hat A\>$ in absence of correlations between the set of modes ${\bf g}$ that are conditioned upon and the set modes ${\bf f}$ for which we construct the reduced state. This result is directly responsible for the previously obtained results related to the spread of non-Gaussian features in cluster states \cite{Walschaers:2018aa}.

\section{Einstein-Podolsky-Rosen steering and Wigner-negativity}\label{sec:steer}

When two systems are connected through a quantum correlation, on can, in some cases, perform quantum steering \cite{RevModPhys.92.015001}. Colloquially, we say that a subsystem ${\cal X}$ can steer a subsystem ${\cal Y}$ when measurements of certain observables in ${\cal X}$ can influence the conditional measurement statistics of observables in ${\cal Y}$ beyond what is possible with classical correlations. Ultimately, in quantum steering one studies properties of conditional quantum states as compared to a local hidden variable model for any observables $X$ and $Y$, acting on ${\cal X}$ and ${\cal Y}$, respectively. Contrary to the case of Bell non-locality, quantum steering considers an asymmetric local hidden variable model:
\begin{equation}
    P(X = x, Y = y) = \sum_\lambda P(\lambda)P(X=x \mid \lambda)P_Q(Y=y \mid \lambda),
\end{equation} 
where one assumes that the probability distributions $P_Q(Y=y \mid \lambda)$ of steered party ${\cal Y}$ follow the laws of quantum mechanics. For the party ${\cal X}$, which performs the steering, no such assumption is made and any probability distribution is allowed. Such local hidden variable model can typically be falsified, either by brute force computational methods \cite{0034-4885-80-2-024001} or via witnesses \cite{PhysRevA.80.032112}. These methods have been applied in a variety of contexts to experimentally observe quantum steering \cite{Saunders:2010aa,PhysRevX.2.031003,Handchen:2012aa,10.1038/ncomms1628,PhysRevLett.110.130407,10.1038/ncomms6886,Cavailles:2018aa,Deng:2017aa,Cai-Steering}.

A paradigmatic example is found when performing homodyne measurements on EPR state \cite{PhysRevLett.60.2731}: when the entanglement in the system is sufficiently strong, one can condition the $\hat x$ and $\hat p$ quadrature measurements in ${\cal Y}$ on the outcome of the same quadrature measurement in ${\cal X}$. The obtained conditional probability distributions for the quadrature measurements in ${\cal Y}$ can violate the Heisenberg inequality, even when averaged over all measurement outcomes in ${\cal X}$. The violations of such a conditional inequality is impossible with classical correlations, but is a hallmark of quantum steering.

Quantum steering can occur in all types of quantum states, with all kinds of measurements. In CV, one often refers to the particular case of Gaussian states that can be steered through Gaussian measurements as EPR steering. Recently, other forms of steering for Gaussian states have been developed under the name of non-classical steering \cite{frigerio2020nonclassical}. In this approach, one checks whether Gaussian measurements in ${\cal X}$ can induce a nonclassical conditional state in ${\cal Y}$. Throughout our current work, the focus lies on EPR steering, where the systems ${\cal X}$ and ${\cal Y}$ are the sets of modes ${\bf f}$ and ${\bf g}$, respectively.\\   


In previous work, we showed EPR steering is a necessary prerequisite to remotely generate Wigner negativity through photon subtraction \cite{PhysRevLett.124.150501}. More precisely, when a photon is subtracted in a mode $g$, the reduced state Wigner function of a correlated mode $f$ can only be non-positive if mode $f$ is able to steer mode $g$. When one allows for an additional Gaussian transformation on mode $g$ prior to photon subtraction, we found that EPR steering from $f$ to $g$ is also a sufficient condition to reach Wigner negativity in mode $f$. 

The formalism that was developed in the previous section allows us to generalize this previous result to arbitrary conditional operations on an arbitrary number of modes:
\begin{theorem}
For any initial Gaussian state $\hat \rho$ and any conditional operation $\hat A$ in (\ref{eq:condState}),  EPR steering between the set of modes $\bf f$ and the set of modes $\bf g$ is {\em necessary} to induce Wigner negativity in $W_{{\bf f} \mid \hat A }(\vec{x}_{\bf f})$.
\end{theorem}
\begin{proof} 
Gaussian EPR steering is generally quantified through the properties of $V_{{\bf g}\mid \vec x_{\bf f}}$. In particular, one can show that the set of modes in ${\bf f}$ can jointly steer the set of modes ${\bf g}$ if and only if $V_{{\bf g}\mid \vec x_{\bf f}}$ violates the Heisenberg inequality \cite{Wiseman:2007aa,Kogias:2015aa}. The crucial consequence is that $W(\vec{x}_{\bf g} \mid \vec{x}_{\bf f})$, as defined in (\ref{eq:condProbDist}), is itself a well-defined Gaussian quantum state when the modes in ${\bf f}$ {\em cannot} steer the modes ${\bf g}$. For all possible $\vec{x}_{\bf f},$ we can thus associate this Gaussian quantum state with a density matrix $\hat \rho_{{\bf g}\mid \vec x_{\bf f}}$.

The crucial observation is that, for any $\vec x_{\bf f}$, $\int_{\mathbb{R}^{2l'}} W_{\hat A}(\vec{x}_{\bf g}) W(\vec{x}_{\bf g} \mid \vec{x}_{\bf f}) {\rm d}\vec{x}_{\bf g}$ is the expectation value of $\hat A$ in a well-defined quantum state $\rho_{{\bf g} \mid \vec x_{\bf f}}$. Because $\hat A$ is a positive semi-definite operator, we directly find that
\begin{equation}
\<\hat A\>_{{\bf g} \mid \vec x_{\bf f}} = \tr [\rho_{{\bf g}\mid \vec x_{\bf f}} \hat A] \geqslant 0 \quad \text{ for all } \quad \vec x_{\bf f} \in \mathbb{R}^{2l}.
\end{equation}
Therefore, the overall conditional Wigner function $W_{{\bf f} \mid \hat A }(\vec{x}_{\bf f})$ in (\ref{eq:Wf}) is non-negative. We can only achieve $\<\hat A\>_{{\bf g} \mid \vec x_{\bf f}} < 0$ for certain points $\vec x_{\bf f} \in \mathbb{R}^{2l}$ when $V_{{\bf g}\mid \vec x_{\bf f}}$ violates the Heisenberg inequality. This concludes that in absence of EPR steering $W_{{\bf f} \mid \hat A }(\vec{x}_{\bf f}) \geqslant 0$.\end{proof}


Note that the steps in this proof rely heavily on the fact that the initial state is Gaussian. For other types of quantum states, we cannot directly relate quantum steering to the properties of $W(\vec{x}_{\bf g} \mid \vec{x}_{\bf f})$.

\section{Examples}\label{sec:examples}

\subsection{Heralding}

In the first example, we consider a scenario where a photon-number revolving measurement is performed on one of the output modes, which can be considered a special case of the situation considered in \cite{su2019generation}. Heralding is ubiquitous in quantum optics, as it is one of the most common tools to generate single photon Fock states \cite{PhysRevLett.56.58,PhysRevLett.87.050402,Bra_czyk_2010}.

To study heralding, we use (\ref{eq:Wf}) where a measurement of the number of photon $n$ in a single mode $g$ is performed. We assume that this measurement is optimal, and, thus, that we project on a Fock state $\lvert n\rangle$. In this case, we set $\hat A = \lvert n\rangle\langle n\rvert$, and therefore we obtain that
\begin{equation}
W_{\hat A}(\vec{x}_g) = \sum _{k=0}^n \binom{n}{k}\frac{(-1)^{n+k} \norm{\vec{x}_g}^{2k}}{k!} \frac{e^{-\frac{1}{2} \norm{\vec{x}_g}^2}}{2 \pi},
\end{equation}
where we used the closed form of the Laguerre polynomial. Hence, we can now use this expression to calculate $\< \lvert n\rangle\langle n\rvert \>_{g \mid \vec x_{\bf f}}$.
It is convenient to explicitly write
\begin{equation}
W(\vec{x}_g \mid \vec{x}_{\bf f})  =  \frac{\exp \left[ -\frac{1}{2} (\vec{x}_g - \vec{\xi}_{g \mid \vec x_{\bf f}})^T V_{g \mid \vec x_{\bf f}}^{-1} (\vec{x}_g - \vec{\xi}_{g \mid \vec x_{\bf f}}) \right]}{2 \pi \sqrt{\det V_{g \mid \vec x_{\bf f}}}},
\end{equation}
and we can recast
\begin{equation}\begin{split}
&\exp \left[ -\frac{1}{2} (\vec{x}_g - \vec{\xi}_{g \mid \vec x_{\bf f}}])^T V_{g \mid \vec x_{\bf f}}^{-1} (\vec{x}_g - \vec{\xi}_{g \mid \vec x_{\bf f}}) -\frac{1}{2} \norm{\vec{x}_g}^2\right]\\
&= e^{ -\frac{1}{2} [(\mathds{1} + V_{g\mid \vec x_{\bf f}})\vec{x}_g - \vec{\xi}_{g \mid \vec x_{\bf f}}]^T [V_{g \mid \vec x_{\bf f}}(\mathds{1}+V_{g \mid \vec x_{\bf f}})]^{-1} [(\mathds{1} + V_{g\mid \vec x_{\bf f}})\vec{x}_g - \vec{\xi}_{g \mid \vec x_{\bf f}}]}\\
&\quad\times e^{-\frac{1}{2} \vec{\xi}_{g\mid \vec x_{\bf f}}^T [\mathds{1} + V_{g\mid \vec x_{\bf f}}]^{-1}\vec{\xi}_{g\mid \vec x_{\bf f}}}.
\end{split}
\end{equation}
After a substitution in the integral, we then find that
\begin{equation}\begin{split}\label{eq:HeraldingComp}
&\< \lvert n\rangle\langle n\rvert \>_{g \mid \vec x_{\bf f}} \\&= 2 \det(\mathds{1} + V_{g\mid \vec x_{\bf f}})^{-1/2} e^{-\frac{1}{2} \vec{\xi}_{g\mid \vec x_{\bf f}}^T [\mathds{1} + V_{g\mid \vec x_{\bf f}}]^{-1}\vec{\xi}_{g\mid \vec x_{\bf f}}}\sum _{k=0}^n \binom{n}{k}\frac{(-1)^{n+k}}{k!}\\
&\quad\times \int_{\mathbb{R}^2}  \frac{ \norm{(\mathds{1} + V_{g\mid \vec x_{\bf f}})^{-1}\vec{x}_g}^{2k}e^{ -\frac{1}{2} (\vec{x}_g - \vec{\xi}_{g \mid \vec x_{\bf f}})^T \sigma^{-1} (\vec{x}_g - \vec{\xi}_{g \mid \vec x_{\bf f}})}}{2 \pi \sqrt{\det \sigma}}{\rm d}\vec{x}_g,
\end{split}
\end{equation}
where we defined $\sigma = V_{g \mid \vec x_{\bf f}}(\mathds{1}+V_{g \mid \vec x_{\bf f}})$, which is now the covariance matrix of a new Gaussian probability distribution. 
The final expression is then determined by the moments of the Gaussian distribution with covariance matrix $\sigma$ and displacement $\vec{\xi}_{g\mid \vec x_{\bf f}}$. Even though this expression is relatively elegant, it can be remarkably tedious to compute for larger values of $n$.\\

First, let us focus on the experimentally relevant case where $n=1$ as an illustration. The evaluation of (\ref{eq:HeraldingComp}) is than conducted by calculating the second moments of a Gaussian distribution, such that we ultimately find
\begin{equation}
\begin{split}
W_{{\bf f} \mid\, \ket{1}\bra{1} }&(\vec{x}_{\bf f}) =\\
&\big[\norm{(\mathds{1} + V_{g\mid \vec x_{\bf f}})^{-1} \vec{\xi}_{g \mid \vec x_{\bf f}}}^2 + \tr [(\mathds{1}+ V_{g\mid \vec x_{\bf f}})^{-1}V_{g\mid \vec x_{\bf f}}] -1\big]\\
&\times  \frac{\det(\mathds{1} + V_{g})^{1/2}}{\det(\mathds{1} + V_{g\mid \vec x_{\bf f}})^{1/2}}  \frac{e^{-\frac{1}{2} \vec{\xi}_{g\mid \vec x_{\bf f}}^T [\mathds{1} + V_{g\mid \vec x_{\bf f}}]^{-1}\vec{\xi}_{g\mid \vec x_{\bf f}}} }{ \tr [(\mathds{1}+ V_{g})^{-1}V_{g}] -1} W_{\bf f} (\vec{x}_{\bf f}),
\end{split}
\end{equation}
where we set $\vec{\xi}_g = 0$, thus assuming that there is no mean field in mode $g$. We note that this function reaches negative values if and only of $\tr [(\mathds{1}+ V_{g\mid \vec x_{\bf f}})^{-1}V_{g\mid \vec x_{\bf f}}] < 1$. Using Williamson's decomposition as we did in \cite{PhysRevLett.124.150501}, it can be shown that this condition can only be fulfilled when the set of modes $\bf f$ can perform EPR steering in mode $g$, or, in other words, when $V_{g\mid \vec x_{\bf f}}$ violates the Heisenberg inequality. This is exactly what we can expect from our general result in Section \ref{sec:steer}.\\ 

 In general we know that the Wigner function (\ref{eq:Wf}) can only be negative when $V_{g\mid\vec x_{\bf f}}$ is not a covariance matrix of a well-defined quantum state. However, determining the existence of zeroes of this Wigner function is a cumbersome task. For heralding with $n > 1$ we, therefore, restrict to numerical simulations using a specific initial state.
 
This specific initial state is generated by mixing two squeezed thermal states on a balance beamsplitter, where one of the output modes will serve as $f$, and the other as $g$. In the limiting case where the initial thermal noise vanishes, we recover the well-known EPR state which manifests perfect photon-number correlations between modes $f$ and $g$. In this case, it is clear a detection of $n$ photons in mode $g$ will herald the state $\ket{n}$ in mode $f$. However, by introducing thermal noise the photon-number correlations fade and the properties of the heralded state in mode $f$ are less clear. Thermal noise will also gradually reduce the EPR steering in the system, such that the Wigner negativity in mode $f$ will vanish when the thermal noise becomes too strong. Hence, with this example we can study the interplay between Wigner negativity and EPR steering in a controlled setting. 

The squeezed thermal state is characterised by a covariance matrix $V= {\rm diag}[\delta/s, \delta s]$, where $\delta$ denotes the amount of initial thermal noise, and $s$ is the squeezing parameter. We initially start with two copies of such a state, and rotate the phase of one of them by $\pi/2$ (see Fig.~\ref{Fig:Heralding}). When both modes are mixed on a beamsplitter, the resulting state manifests EPR steering depending on parameters $\delta$ and $s$, which can be quantified through \cite{Kogias:2015aa}
 \begin{equation}\label{eq:mu}
     \mu = \max \left\{0, -\frac{1}{2}\log \det V_{g\mid \vec x_f}\right\},
 \end{equation}
where we explicitly use the fact that $V_{g\mid \vec x_f}$ is a two-dimensional matrix. When we then post-select on the number of photons, $n$, measured in one output mode, we herald a conditional non-Gaussian state in the other mode. In Fig.~\ref{Fig:Heralding}, we show the resulting Wigner functions for the case where the detected number of photons is $n=5$. When the amount of EPR steering is varied (note that $\mu = 0.55$ corresponds to the pure state), we see that the resulting Wigner function rapidly loses Wigner negativity. In full agreement with our general result of the previous section, we also find that the Wigner negativity vanishes when there is no EPR steering. 

A more quantitative study of the Wigner negativity can be found in Panel (c) of Fig.~\ref{Fig:Heralding}, where we vary, both, the amount of steering $\mu$ and the number of detected photons $n$. The Wigner negativity is measured by the quantity \cite{Takagi:2018aa,PhysRevA.98.052350,Kenfack:2004aa}
\begin{equation}\label{eq:wigneg}
{\cal N} = \int_{\mathbb{R}^2} \abs{W_{{f} \mid \hat A }(\vec{x}_{f})} {\rm d}\vec x_{f} - 1 
\end{equation}
When the state is pure (here for $\mu = 0.55$), a detection of $n$ photons in one mode herald a Fock state $\ket{n}$ in the other mode and the Wigner negativity, thus, increases with $n$. However, once the state is no longer pure and the steering decreases, we observe the existence of an optimal value $n$ for which the maximal amount of Wigner negativity is obtained. For very weak EPR steering (e.g. $\mu =0.08$ in this calculation), this optimal value is obtained for $n=1$.
 
\begin{figure*}
\includegraphics[width=0.99\textwidth]{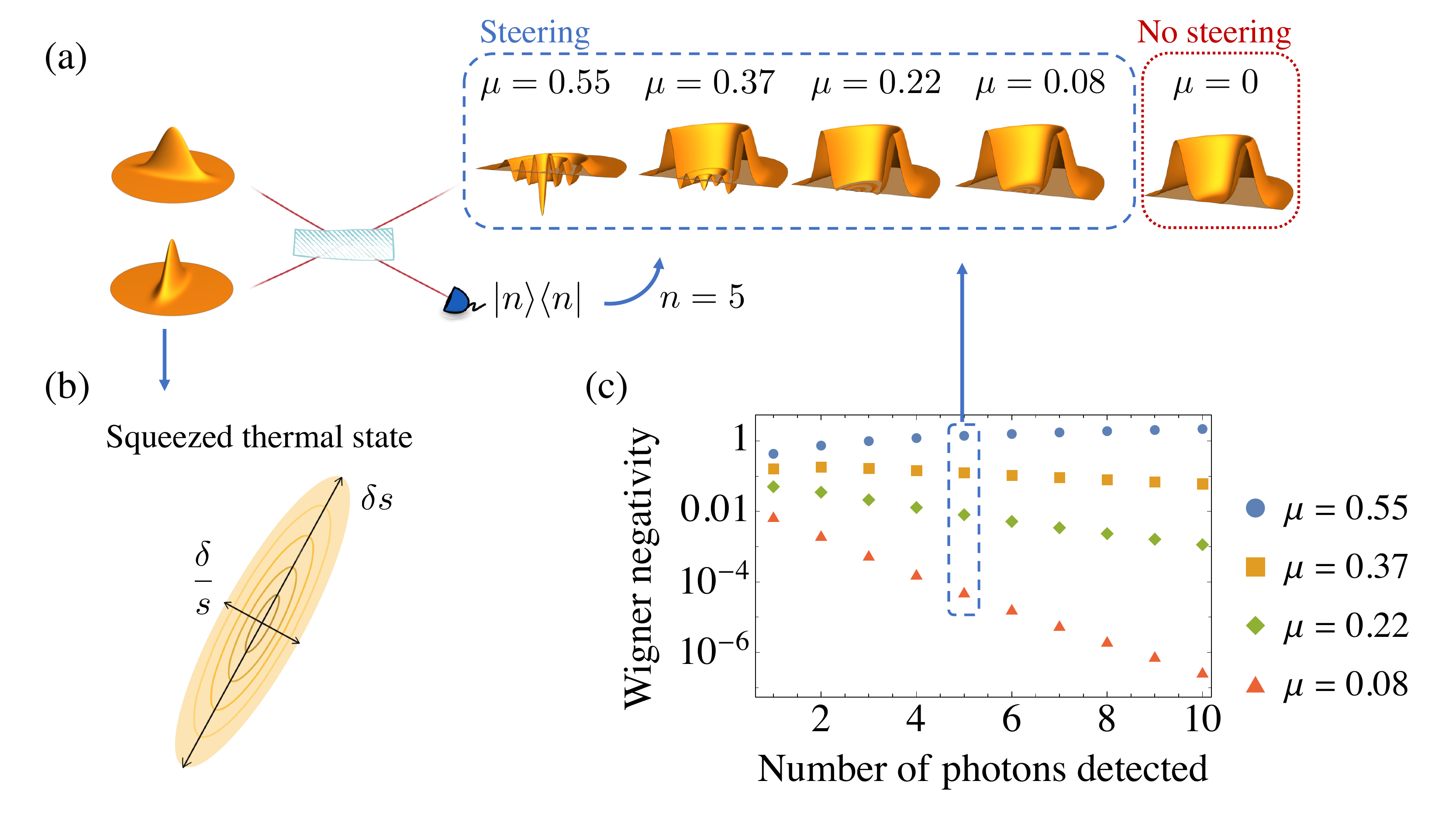}
\caption{Photon heralding with a particular Gaussian input state, generated by mixing two equal squeezed thermal states (b) on a balanced beamsplitter (a). On one of the outputs of the beamsplitter, a projective measurement is performed on the Fock state $\ket{n}$, which heralds a non-Gaussian state in the other mode. The Wigner functions of this non-Gaussian state are shown for the case where $n=5$, with varying degrees of EPR steering $\mu$, controlled by varying the thermal noise $\delta$ for a fixed squeezing $s=5{\rm dB}$. The Wigner negativity, measured by ${\cal N}$ (\ref{eq:wigneg}) is shown in (c) for varying degrees of EPR steering and a varying number of measured photons $n$. \label{Fig:Heralding}}
\end{figure*}

This numerical study shows the fruitfulness of our presented framework to study a very concrete heralding scheme. Furthermore, the example confirms the relationship between Gaussian EPR steering and Wigner negativity.

\subsection{Photon-added and -subtracted states}

Ideal photon addition and subtraction are defined by acting with a creation operator $\hat a^{\dag}$ or annihilation operator $\hat a$, respectively, on the quantum state. In practice, these operations are often realised by using some form of heralding \cite{Wenger:2004aa}, which we treated in the previous example. However, it tends to be more convenient to use the idealised model, based on creation and annihilation operators, and it has been shown experimentally that this model is highly accurate. This model also fits the conditional state framework of (\ref{eq:condFormal}), where we set $\hat X_j$ to be a creation or annihilation operator.

In multimode systems, photon addition and subtraction have been considered for their entanglement properties, which sprouted a range of theoretical \cite{PhysRevA.61.032302,PhysRevA.67.032314,PhysRevLett.93.130409,PhysRevA.73.042310,PhysRevA.80.022315,PhysRevA.86.012328,PhysRevA.93.052313,walschaers_entanglement_2017} and experimental \cite{ourjoumtsev_increasing_2007,takahashi_entanglement_2010,Morin:2014aa} results. Many of the obtained theoretical results rely on the purity of the initial Gaussian state, and are hard to generalise to arbitrary Gaussian states. In recent years, there has been some progress in developing analytical tools to describe general photon subtracted states \cite{walschaers_entanglement_2017,PhysRevA.99.053816}, but it remain challenging to use these techniques to evaluate entanglement measures. Therefore, one has investigated related questions, such as, for example, the spread of non-Gaussian features in multimode systems \cite{Fedorov:15,Katamadze:18,Walschaers:2018aa,PhysRevLett.124.150501}.

The framework presented in this manuscript is particularly fruitful to investigate the spread of non-Gaussian features through photon addition or subtraction. We will first show how the results of \cite{PhysRevLett.124.150501} can be recovered via (\ref{eq:Wf}). Then, we use the present framework to provide analytical results for the states that can be obtained by subtracting multiple photons in a multimode system.

\subsubsection{Adding or subtracting a single photon}

We will start by studying the addition and subtract of a single photon. The scenario for photon-subtracted states was studied in detail in \cite{PhysRevLett.124.150501} and out goal in this example is to show how these previous results can be obtain in the context of our present framework. Furthermore, we also study photon addition, which not yet been considered in the context of the remote generation of Wigner negativity.  

Creating and annihilation operators are by construction operators that act on a single mode $g$. In the single photon scenario, we find the photon-subtracted state
\begin{equation}
\hat \rho_- = \frac{ \hat a_g\hat \rho \hat a_g^{\dag}}{\tr [\hat n_g \hat\rho]},
\end{equation}
and the photon-added state
\begin{equation}
\hat \rho_+ = \frac{ \hat a^{\dag}_g\hat\rho \hat a_g}{\tr [(\hat n_g + \mathds{1}) \hat\rho]}.
\end{equation}
These states are clearly fit the framework of (\ref{eq:condFormal}). In the context of (\ref{eq:condState}), the reduced state of the set of modes $\bf f$ is obtained by choosing $\hat A = \hat n_g$ and $\hat A= \hat n_g + \mathds{1}$ for photon subtraction and addition, respectively.  We can then use (\ref{eq:Wf}) to obtain the Wigner function in the subset of modes $\bf f$, for which we must evaluate $\<\hat n_g\>_{g \mid \vec x_{\bf f}}$. To this goal, we evaluate the Wigner function of the number operator, which is given by
\begin{equation}
W_{\hat n_g}(\vec{x}_g) = \frac{1}{16 \pi} (\norm{\vec{x}_g}^2 - 2),
\end{equation}
such that we directly find that
\begin{equation}\label{eq:En}
\<\hat n_g\>_{g \mid \vec x_{\bf f}} = \frac{1}{4}(\tr V_{g \mid \vec x_{\bf f}} + \norm{\vec\xi_{g\mid \vec x_{\bf f}}}^2 - 2),
\end{equation}
where the dependence on $\vec x_{\bf f}$ comes from $\vec\xi_{g\mid \vec x_{\bf f}}$. 

Thus, we find for the photon-subtracted state that 
\begin{equation}
W^-_{{\bf f} \mid \hat n_g}(\vec{x}_{\bf f}) = \frac{\tr V_{g \mid \vec x_{\bf f}} + \norm{\vec\xi_{g\mid \vec x_{\bf f}}}^2 - 2}{\tr V_{g} + \norm{\vec\xi_{g}}^2 - 2}W_{\bf f}(\vec{x}_{\bf f}).
\end{equation}
From this result, we immediately observe that the potential Wigner-negativity of these states depends on whether or not $\tr V_{g \mid \vec x_{\bf f}} < 2$. In \cite{PhysRevLett.124.150501} it was shown through the Williamson decomposition that $\tr V_{g \mid \vec x_{\bf f}} \leqslant 2 \sqrt{\det V_{g\mid \vec x_f}}$. This directly implies that EPR steering (\ref{eq:mu}) is a necessary condition to reach Wigner-negativity. It is instructive to emphasise that 
\begin{equation}
\norm{\vec\xi_{g \mid \vec x_{\bf f}}}^2 = \norm{\vec{\xi}_g + V_{g{\bf f}} V_{\bf f}^{-1} (\vec{x}_{\bf f} - \vec{\xi}_{\bf f})}^2,
\end{equation} 
from which one ultimately retrieves the expression
\begin{align}\label{eq:WignerDispRed}
W_{{\bf f} \mid \hat n_g}^{-}(\vec{x}_{\bf f})= \frac{\Big\{ \norm{\vec{\xi}_g + V_{g{\bf f}} V_{\bf f}^{-1} (\vec{x}_{\bf f} - \vec{\xi}_{\bf f})}^2  +\tr V_{g\mid \vec x_{\bf f}} - 2\Big\}}{\tr V_g + \norm{\vec\xi_g}^2 - 2} W_{\bf f}(\vec{x}_{\bf f}) \nonumber,
\end{align}
which is the result that was derived in \cite{PhysRevLett.124.150501}.\\

For the photon-added state, we can perform a completely analogous computation with \begin{equation}
W_{\hat n_g + \mathds{1}}(\vec{x}_g) = \frac{1}{16 \pi} (\norm{\vec{x}_g}^2 + 2),
\end{equation} from which we find that
\begin{equation}\label{eq:photonAdded}
W^+_{{\bf f} \mid \hat n_g}(\vec{x}_{\bf f}) = \frac{\tr V_{g \mid \vec x_{\bf f}} + \norm{\vec\xi_{g\mid \vec x_{\bf f}}}^2 + 2}{\tr V_{g} + \norm{\vec\xi_{g}}^2 + 2}W_{\bf f}(\vec{x}_{\bf f}).
\end{equation}
This result immediately shows that this Wigner function is always positive, which implies that it is impossible to remotely create Wigner negativity through photon addition.

In previous work, we highlighted that photon addition always creates Wigner negativity in the mode where the photon is added \cite{walschaers_entanglement_2017}. What we observe in (\ref{eq:photonAdded}) can be understood as the complementary picture for the other modes. This result also highlights an operational difference between photon subtraction and addition: photon additional is a more powerful tool to locally create Wigner negativity, whereas photon subtraction has the potential to create Wigner negativity non-locally (i.e. in modes that can steer the mode in which the photon is subtracted). 

\subsubsection{Subtracting multiple photons}\label{sec:multiphoton}

When multiple photons are added or subtracted, or when we chain combinations of addition and subtraction operations, the evaluation of $\<\hat A\>_{{\bf g} \mid \vec x_{\bf f}}(\vec{x}_{\bf f})$ will rapidly become more complicated. A general strategy to approach this problem avoids the explicit evaluation of $W_{\hat A}(\vec{x}_{\bf g})$, but rather uses  standard techniques for the evaluation of moments of multivariate Gaussian distributions. This ultimately boils down to applying Wick's theorem \cite{PhysRev.80.268} and summing over all matchings (see Appendix \ref{app:Matching} for details). Even though this task can be implemented numerically, the corresponding analytical expressions quickly become intractable.

To illustrate this method, we consider the multimode scenario where two photons are subtracted in different orthogonal modes, $g_1$ and $g_2$, which implies that the conditioning implements the following map
\begin{equation}
    \rho \mapsto \frac{\hat a_{g_1} \hat a_{g_2} \hat \rho \hat a^{\dag}_{g_2}\hat a^\dag_{g_1}}{\tr[\hat n_{g_1} \hat n_{g_2} \hat \rho]}.
\end{equation}
This implies that we must apply our formalism with $\hat A = \hat n_{g_1} \hat n_{g_2}$. To treat this problem with the technique of matchings, we use the Gaussian identity (not that we do not explicitly write the dependence on $\vec x_{\bf f}$ to simplify notation) 
\begin{equation}\label{eq:long}
    \begin{split}
    \gf{\hat n_{g_1} \hat n_{g_2}} =& \abs{\gf{\hat a_{g_1}}}^2 \abs{\gf{\hat a_{g_2}}}^2+\gf{\hat n_{g_1}}' \abs{\gf{\hat a_{g_2}}}^2\\
    &+\gf{\hat n_{g_2}}' \abs{\gf{\hat a_{g_1}}}^2+\gf{\hat a^{\dag}_{g_1}\hat a_{g_2}}' \gf{\hat a^{\dag}_{g_2}}\gf{\hat a_{g_1}}\\
    &+\gf{\hat a^{\dag}_{g_1}\hat a_{g_2}}' \gf{\hat a^{\dag}_{g_2}\hat a_{g_1}}'+\gf{\hat a^{\dag}_{g_1}\hat a^{\dag}_{g_2}}' \gf{ \hat a_{g_1} \hat a_{g_2}}'\\
    &+ \gf{\hat n_{g_1}}' \gf{\hat n_{g_2}}' +\gf{\hat a^{\dag}_{g_2}\hat a_{g_1}}' \gf{a^{\dag}_{g_1}}\gf{\hat a_{g_2}}\\
    &+\gf{\hat a^{\dag}_{g_1}\hat a^{\dag}_{g_2}}' \gf{\hat a_{g_1}}\gf{\hat a_{g_2}}\\
    &+\gf{\hat a_{g_1}\hat a_{g_2}}' \gf{\hat a^{\dag}_{g_1}}\gf{\hat a^{\dag}_{g_2}},
    \end{split}
\end{equation}
where $\gf{\dots}'$ denotes the non-displaced version of the distribution. We can immediately identify \begin{equation}\gf{\hat a_{g_1}} = \frac{1}{2} (\vec{\xi}_{{\bf g} \mid \vec x_{\bf f}}^T \vec g_1 + i \vec{\xi}_{{\bf g} \mid \vec x_{\bf f}}^T \Omega \vec g_1),\end{equation}
subsequently, we obtain from (\ref{eq:En}) that
\begin{equation}\gf{\hat n_{g_1}}' = \frac{1}{4}(\tr V_{g_1 \mid \vec x_{\bf f}} - 2) ,\end{equation}
and finally we find new types of terms, that are given by
\begin{equation}\begin{split}\gf{\hat a^{\dag}_{g_1} \hat a^{\dag}_{g_2}}' = \frac{1}{4}[&\vec g_1^T V_{{\bf g} \mid \vec x_{\bf f}} \vec g_2 - \vec g_1^T \Omega^T V_{{\bf g} \mid \vec x_{\bf f}}\Omega \vec g_2\\ & -i ( \vec g_1^T V_{{\bf g} \mid \vec x_{\bf f}} \Omega \vec g_2 + \vec g_1^T \Omega^T V_{{\bf g} \mid \vec x_{\bf f}} \vec g_2)] ,
\end{split}\end{equation}
and
\begin{equation}\begin{split}\gf{\hat a^{\dag}_{g_1} \hat a_{g_2}}' = \frac{1}{4}[&\vec g_1^T V_{{\bf g} \mid \vec x_{\bf f}} \vec g_2 + \vec g_1^T \Omega^T V_{{\bf g} \mid \vec x_{\bf f}}\Omega \vec g_2\\ & + i ( \vec g_1^T V_{{\bf g} \mid \vec x_{\bf f}} \Omega \vec g_2 - \vec g_1^T \Omega^T V_{{\bf g} \mid \vec x_{\bf f}} \vec g_2)].
\end{split}\end{equation}
Computation required to obtain the final result is tedious but straightforward. We find that 
\begin{equation}\label{eq:En1n2}
    \begin{split}
    \gf{\hat n_{g_1} \hat n_{g_2}} = \frac{1}{16}[&(\tr V_{g_1\mid \vec x_{\bf f}} + \norm{\vec \xi_{g_1 \mid \vec x_{\bf f}}}^2-2)(\tr V_{g_2\mid \vec x_{\bf f}} + \norm{\vec \xi_{g_2 \mid \vec x_{\bf f}}}^2-2)\\
    &+ 2 \tr(C^TC) + 4 \vec \xi_{g_1 \mid \vec x_{\bf f}}^T C \vec \xi_{g_2 \mid \vec x_{\bf f}} ],
    \end{split}
\end{equation}
where we have defined the submartix $C$ as the off-diagonal block of $V_{{\bf g} \mid \vec x_{\bf f}}$ via
\begin{equation}
V_{{\bf g} \mid \vec x_{\bf f}} =
    \begin{pmatrix}
    V_{g_1 \mid \vec x_{\bf f}} & C \\
    C^T & V_{g_2 \mid \vec x_{\bf f}}
    \end{pmatrix}.
\end{equation}
Non-zero entries in the block $C$ can occur due to various caused. First of all, it can be due to a correlation between the modes $g_1$ and $g_2$ in the initial Gaussian state (as seen from the term $V_{\bf g}$ in (\ref{eq:schur})). However, non-trivial entries in $C$ also arise when modes $g_1$ and $g_2$ are both correlated to the same modes in ${\bf f}$, which is induced by the term $V_{{\bf gf}}V_{\bf f}^{-1}V_{\bf gf}^T$ in (\ref{eq:schur}).

The result (\ref{eq:En1n2}) directly show the appearance of a trivial term, $(\tr V_{g_1\mid \vec x_{\bf f}} + \norm{\vec \xi_{g_1 \mid \vec x_{\bf f}}}^2-2)(\tr V_{g_1\mid \vec x_{\bf f}} + \norm{\vec \xi_{g_1 \mid \vec x_{\bf f}}}^2-2)$, which multiplies the effect of photon subtraction in $g_1$ with that of photon subtraction in $g_2$. However, when both modes are sufficiently ``close'' to each other, we find the additional terms $2 \tr(C^TC) + 4 \vec \xi_{g_1 \mid \vec x_{\bf f}}^T C \vec \xi_{g_2 \mid \vec x_{\bf f}}$, which can be interpreted as some form of interference between the two photon subtractions.

Fig.~\ref{Fig:Two-Photon-Sub} provides an illustration, where we inject three pure squeezed vacuum states into a series of beamsplitters to generate an entangled three-mode state from which we subtract two photons. The first two squeezed vacuum states have $5 {\rm dB}$ squeezing in opposite quadratures and are mixed on beamsplitter with $75\%$ transmittance. One of the output ports will serve as mode $g_1$, whereas the other is injected into a section beamsplitter of $25\%$ transmittance. In the other input port of this beamplitter, we inject the third squeezed vacuum state, which is also squeezed by $5 {\rm dB}$. One of the output ports of the $25\%$ transmittance beamsplitter serves as mode $g_2$, and in the other output port we find mode $f$, which is the mode for which we reconstruct the output Wigner function using (\ref{eq:En1n2}). Photon subtraction is represented by a highly transmitting beamsplitter which sends a small amount of light to a photon detector. Two-photon subtraction then happens when both detectors click at the same time, and we can condition the state in mode $f$ upon this detection outcome. This post-selection scheme effectively implements the operators $\hat a_{g_1}$ and $\hat a_{g_2}$ on modes $g_1$ and $g_2$, respectively.

We observe that the conditional state $W_{f\mid \hat A}(\vec{x}_f)$, with $\hat A = \hat n_{g_1} \hat n_{g_2}$ reaches negative values in two distinct regions of phase space. Indeed, with the Williamson decomposition of $V_{{\bf g} \mid \vec x_{f}}$ we can quantify \cite{Kogias:2015aa} the strength of EPR steering from mode $f$ to the set of modes ${\bf g}$ to be $\mu=0.548$. Furthermore, the fact that there are two negativity regions is a hallmark of the subtraction of two photons. This example shows that our framework is a highly versatile tool for CV quantum state engineering. \\

\begin{figure}
\includegraphics[width=0.49\textwidth]{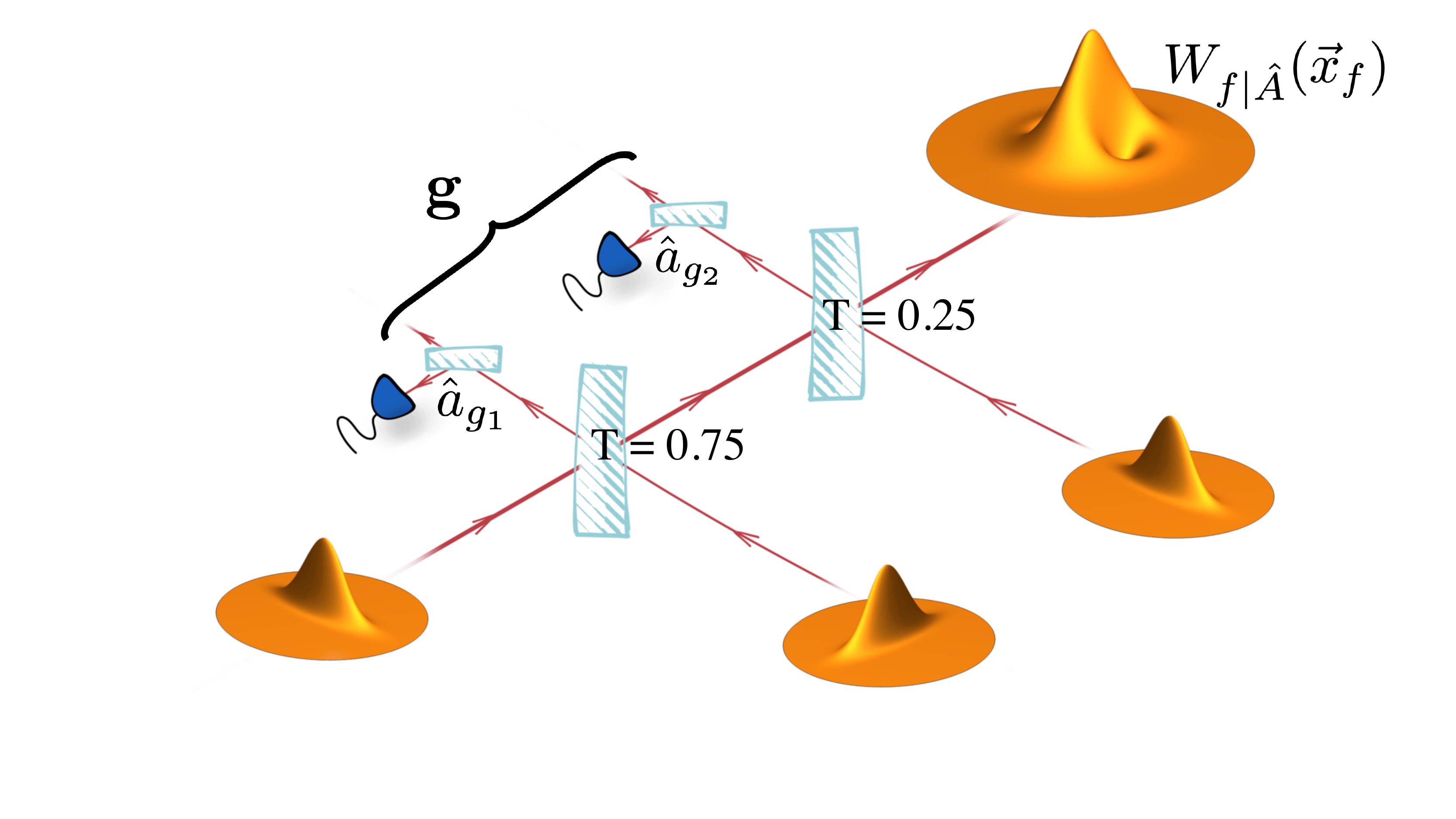}
\caption{Conditional state Wigner function $W_{f\mid \hat A}(\vec{x}_f)$, obtained by subtracting a photon in two of the three modes in a three-mode entangled state. This entangled state is generated by mixing three squeezed vacuum states in a sequence of beamsplitters with transmittance of $75\%$ (left) and $25\%$ (right). Two of the squeezed vacuum states are squeezed by $5 {\rm dB}$ in the $x$-quadrature (left, right) and one is squeezed by $5 {\rm dB}$ along the $p$ quadrature (middle). The photon subtraction is represented by highly transmitting beamsplitters which send a small fraction of light to a photon detector, which effectively implements the operators $\hat a_{g_1}$ and $\hat a_{g_2}$ on the mode $g_1$ and $g_2$, respectively.  \label{Fig:Two-Photon-Sub}}
\end{figure}

Finally, we consider the complementary scenario where two photons are subtracted from one mode. In this case, we can still use the perfect matching technique (\ref{eq:long}), when creation and annihilation operators are in normal ordering. In this case, we obtain $\hat A = \hat a^\dag_g  \hat a^\dag_g  \hat a_g  \hat a_g,$ and analogously to (\ref{eq:long}), we find that 
\begin{equation}\begin{split}
    \<\hat a^\dag_g  \hat a^\dag_g  \hat a_g  \hat a_g\>_{g \mid \vec x_{\bf f}} = \frac{1}{16}\Big[& (\tr V_{g\mid \vec x_{\bf f}} + \norm{\vec \xi_{g \mid \vec x_{\bf f}}}^2)^2 + 2 \tr(V_{g\mid \vec x_{\bf f}}^2)\\
    &+ 4 \vec{\xi}^T_{g\mid \vec x_{\bf f}} [V_{g\mid \vec x_{\bf f}} - 2 \mathds{1} ] \vec{\xi}_{g\mid \vec x_{\bf f}}\\
    &- 8\tr V_{g \mid \vec x_{\bf f}} + 8\Big].
    \end{split}
\end{equation}
This result can then directly be inserted in (\ref{eq:Wf}) to obtain the final conditional state for the set of modes ${\bf f}$ when two photons are subtracted in mode $g$. As expected, the subtraction of two photons can induce Wigner negativity only when there is EPR steering from the modes ${\bf f}$ to mode $g$. \\

As such, we have shown that our framework allows us to analytically describe conditional non-Gaussian states in a regime which is highly challenging for many other methods. For example, it is highly challenging to approach the problem with the correlation function methods of \cite{walschaers_entanglement_2017}, even though this method was highly successful for single-photon subtraction in multimode states.

These methods can in principle be extended to deal with higher numbers of added and/or subtracted photons in various modes. However, it must be emphasised that one will quickly encounter practical boundaries as finding all possible matchings is a computationally hard problem \cite{Matching-Book}. Finding an exact description of the Wigner function that is obtained by subtracting a large number of photons from a subset of an entangled Gaussian state seems to be a computationally hard problem that has its roots in graph theory. The problem of finding all matchings also lies at the basis of Gaussian boson sampling \cite{PhysRevLett.119.170501,PhysRevA.98.032310}, and it is not expected to be easy to overcome. The problem of Gaussian boson sampling can in turn also be related to CV sampling from photon-added or -subtracted states \cite{PhysRevA.96.062307}.

\section{Conclusions}

We presented a general framework that describes the Wigner function that is obtained by applying an arbitrary operation on a subset of modes of a multimode Gaussian state, and conditioning the remaining modes on this operation. The most natural way of interpreting this scenario is by considering this operation to be a measurement, such that the state of the remaining modes is obtained by post-selecting on a specific measurement outcome, as is the case for heralding. However, this framework can also be used to study the non-local effects of photon addition and subtraction.

Our framework relies heavily on classical probability theory, and in particular on properties of conditional probability distributions (\ref{eq:condProbDist}). We use the fact that Gaussian states have positive Wigner functions, such that associated conditional probability distributions on phase space are well-defined as probability distribution (but not necessarily as quantum states, because they can violate the Heisenberg inequality). In this regard, our general results (\ref{eq:EA} - \ref{eq:expA}) are valid for all initial states with a positive Wigner function.\\

Gaussian states are not only the most relevant initial states from an experimental point of view, they also have the theoretical advantage of leading to a Gaussian conditional probability distribution. The latter is an enormous advantage for evaluating the crucial quantity $\<\hat A\>_{{\bf g} \mid \vec x_{\bf f}}$, as defined in (\ref{eq:EA}). On a more fundamental level, we note that the covariance matrix (\ref{eq:schur}) of this Gaussian conditional probability distribution is essential in the theory of Gaussian EPR steering. This observation allows us to directly prove that Gaussian EPR steering is a necessary prerequisite for the conditional preparation of Wigner-negativity, regardless of the conditional operation that is performed. 

In previous work, we already showed that Gaussian EPR steering is also a sufficient ingredient for the remote preparation of Wigner-negativity, in the sense that there always exists a combination of a Gaussian operation and photon subtraction in the modes ${\bf g}$ that induces Wigner-negativity in the modes ${\bf f}$. We thus establish a fundamental relation between Gaussian EPR steering and the ability to prepare a Wigner-negative state in correlated modes. This result is particularly important in the light of measurement-based quantum computation, where large Gaussian cluster states form the backbone for implementing a quantum algorithm. The actual computation is then executed by performing measurements (or more general operations) on some modes of the cluster, in order to project the remainder of the system in a desired quantum state. To claim that such a computation is universal, one must be able to induce Wigner-negativity. Our results therefore show that EPR steering is an essential figure of merit in these cluster states in order to claim that a cluster state is suitable for universal quantum computation.\\

From a practical point of view, the examples in Section \ref{sec:examples} show that our framework is highly versatile. However, it also highlights the boundaries of analytical treatments. Even though the obtained expression (\ref{eq:EA}) for $\<\hat A\>_{{\bf g} \mid \vec x_{\bf f}}$ is easy to interpret conceptually, the actual evaluation can still be challenging. Regardless, we must emphasise that the elegance and simplicity of our framework does allow us to obtain results with far greater ease than previously possible. Many of the methods known in literature are either hard to generalise to arbitrary Gaussian initial states \cite{PhysRevA.72.033822,Thekkadath2020engineering}, focused on one particular measurement or operation \cite{walschaers_entanglement_2017,su2019generation,PhysRevA.99.053816}, or are just generally hard to interpret or use analytically. Our framework can be applied to any initial Gaussian state, and any conditional operation, provided the Wigner function of $\hat A$ is known. 

As such, our results provide a starting point for investigating a wide range of new questions related to multimode conditional preparation of non-Gaussian states. By establishing a fundamental relation between EPR steering and Wigner-negativity, we specifically highlight that this framework is also suited to obtain general analytical results, which is often challenging in the study of states that are, both, highly non-Gaussian and highly multimode.

\begin{acknowledgments}
V.P. acknowledges financial support from the European Research Council under the Consolidator Grant COQCOoN
(Grant No. 820079)
\end{acknowledgments}

\appendix
\section{Matchings}\label{app:Matching}
In Section \ref{sec:multiphoton}, we refer to the method of perfect matchings to evaluate $\gf{\hat A}$, which we here present with more rigour and detail.

The technique of perfect matchings is a common practice to evaluate correlation functions in Gaussian states, which can be traced back to works such as \cite{PhysRev.80.268,robinson_ground_1965}. In formal terms, we consider a Gaussian (also known as ``quasi-free'' in the mathematical physics literature) functional $\gf{\dots}$ on the algebra of observables for the canonical commutation relations \cite{verbeure_many-body_2011}. A defining property of such functionals is that truncated correlation functions \cite{robinson_ground_1965} for any product of more than two creation and annihilation operators vanishes. This property is the direct analog of the cumulants of a multivariate Gaussian distribution and it implies that the functional $\gf{\dots}$ is fully determined by the quantities $\gf{a^{\dag}_{g_1}a^{\dag}_{g_2}}' = \gf{a_{g_1}a_{g_2}}'^*$, $\gf{a^{\dag}_{g_1}a_{g_2}}'$, and $\gf{a^{\dag}_{g_1}} = \gf{a_{g_1}}^*$. Where the $\gf{\dots}'$ is the non-displaced version of the functional, which is formally defined as
\begin{equation}
 \gf{a^{\#}_{g_1}a^{\#}_{g_2}}' = \gf{a^{\#}_{g_1}a^{\#}_{g_2}} - \gf{a^{\#}_{g_1}}\gf{a^{\#}_{g_2}},
 \end{equation}
where $a^{\#}_{g_1}$ can be either a creation or an annihilation operator. We can then write the following general property of Gaussian functional:
\begin{equation}\begin{split}\label{eq:AllMatchings}
    &\gf{a^{\dag}_{g_1}\dots a^{\dag}_{g_n}a_{g_{n+1}}\dots a_{g_{n+m}}}\\
    &= \sum_{M \in {\cal M}}\prod_{\{j_1,j_2\} \in M}  \gf{a^{\#}_{j_1}a^{\#}_{j_2}}' \prod_{\{k\} \in M}\gf{a^{\#}_{k}}.
    \end{split}
\end{equation}
Where ${\cal M}$ is the set of all ``matchings'' for the set $\{g_1, \dots g_{n+m}\}$. We use the term matching to refer to a partition of the set $\{g_1, \dots g_{n+m}\}$ in subsets with either one or two elements. An example of such a possible matching is given by $M = \{\{g_1,g_2\}, \dots,\{g_{n-1},g_{n}\},\{g_n\},  \dots, \{g_{n+m}\}\}$. For each partition $M \in {\cal M}$, we then evaluate the product of associated two-point and one-point functions, where any pair $\{j_1,j_2\} \in M$ is associated with $\gf{a^{\#}_{j_1}a^{\#}_{j_2}}'$ and $\{k\} \in M$ is associated with $\gf{a^{\#}_{k}}$. Note that for $i = g_1, \dots,  g_n$, the operator $a^{\#}_{i}$ is a creation operator, whereas for $i = g_{n+1}, \dots,  g_{n+m}$ it is an annihilation operator. 

The problem of finding all matchings is a well-known problem in graph theory. To make the connection, we can represent each elements of the set $\{g_1, \dots g_{n+m}\}$ as a vertex in a full connected graph, and then consider the resulting partitions as the matchings of this graph \cite{Matching-Book}. The number of terms in (\ref{eq:AllMatchings}) quickly explodes as the number of creation and annihilation operators increases, which ultimately makes the problem of evaluating $\gf{a^{\dag}_{g_1}\dots a^{\dag}_{g_n}a_{g_{n+1}}\dots a_{g_{n+m}}}$ computationally hard. \\   

A subtle point in our treatment of $\gf{\hat A}$ is that $\gf{\dots}$ is not an expectation value of a Gaussian quantum state. Hence, it is legitimate to wonder up to what extent the techniques of Gaussian quantum states can be used to evaluate $\gf{\hat A}$. From its definition in (\ref{eq:EA}), it can be deduced that $\gf{\dots}$ is a functional on the algebra of observables. It directly inherits the Gaussian properties from $W(\vec{x}_{\bf g} \mid \vec{x}_{\bf f})$, such that it is a Gaussian functional. In particular the property (\ref{eq:AllMatchings}) can directly be traced back to the structure of the moments of the multivariate Gaussian probability distribution $W(\vec{x}_{\bf g} \mid \vec{x}_{\bf f})$. The Gaussian functional $\gf{\dots}$ is not associated to a state because it is not a positive functional, i.e., we can find positive operators $\hat A$ for which $\gf{\hat A} < 0$. For a Gaussian functional on the algebra of canonical commutation relations to be equivalent to a quantum state, one must impose additional constraints on the functional to guarantee positivity \cite{robinson_ground_1965,verbeure_many-body_2011}. These constraints ultimately boil down to imposing the Heisenberg inequality.

\bibliography{notes_steering}

\begin{thebibliography}{101}%
\makeatletter
\providecommand \@ifxundefined [1]{%
 \@ifx{#1\undefined}
}%
\providecommand \@ifnum [1]{%
 \ifnum #1\expandafter \@firstoftwo
 \else \expandafter \@secondoftwo
 \fi
}%
\providecommand \@ifx [1]{%
 \ifx #1\expandafter \@firstoftwo
 \else \expandafter \@secondoftwo
 \fi
}%
\providecommand \natexlab [1]{#1}%
\providecommand \enquote  [1]{``#1''}%
\providecommand \bibnamefont  [1]{#1}%
\providecommand \bibfnamefont [1]{#1}%
\providecommand \citenamefont [1]{#1}%
\providecommand \href@noop [0]{\@secondoftwo}%
\providecommand \href [0]{\begingroup \@sanitize@url \@href}%
\providecommand \@href[1]{\@@startlink{#1}\@@href}%
\providecommand \@@href[1]{\endgroup#1\@@endlink}%
\providecommand \@sanitize@url [0]{\catcode `\\12\catcode `\$12\catcode
  `\&12\catcode `\#12\catcode `\^12\catcode `\_12\catcode `\%12\relax}%
\providecommand \@@startlink[1]{}%
\providecommand \@@endlink[0]{}%
\providecommand \url  [0]{\begingroup\@sanitize@url \@url }%
\providecommand \@url [1]{\endgroup\@href {#1}{\urlprefix }}%
\providecommand \urlprefix  [0]{URL }%
\providecommand \Eprint [0]{\href }%
\providecommand \doibase [0]{http://dx.doi.org/}%
\providecommand \selectlanguage [0]{\@gobble}%
\providecommand \bibinfo  [0]{\@secondoftwo}%
\providecommand \bibfield  [0]{\@secondoftwo}%
\providecommand \translation [1]{[#1]}%
\providecommand \BibitemOpen [0]{}%
\providecommand \bibitemStop [0]{}%
\providecommand \bibitemNoStop [0]{.\EOS\space}%
\providecommand \EOS [0]{\spacefactor3000\relax}%
\providecommand \BibitemShut  [1]{\csname bibitem#1\endcsname}%
\let\auto@bib@innerbib\@empty
\bibitem [{\citenamefont {Schr{\"o}dinger}(1926)}]{Schrodinger-Coherent}%
  \BibitemOpen
  \bibfield  {author} {\bibinfo {author} {\bibfnamefont {E.}~\bibnamefont
  {Schr{\"o}dinger}},\ }\href {\doibase 10.1007/BF01507634} {\bibfield
  {journal} {\bibinfo  {journal} {Naturwissenschaften}\ }\textbf {\bibinfo
  {volume} {14}},\ \bibinfo {pages} {664} (\bibinfo {year} {1926})}\BibitemShut
  {NoStop}%
\bibitem [{\citenamefont {Kennard}(1927)}]{Kennard}%
  \BibitemOpen
  \bibfield  {author} {\bibinfo {author} {\bibfnamefont {E.~H.}\ \bibnamefont
  {Kennard}},\ }\href {\doibase 10.1007/BF01391200} {\bibfield  {journal}
  {\bibinfo  {journal} {Zeitschrift f{\"u}r Physik}\ }\textbf {\bibinfo
  {volume} {44}},\ \bibinfo {pages} {326} (\bibinfo {year} {1927})}\BibitemShut
  {NoStop}%
\bibitem [{\citenamefont {Glauber}(1963)}]{PhysRev.131.2766}%
  \BibitemOpen
  \bibfield  {author} {\bibinfo {author} {\bibfnamefont {R.~J.}\ \bibnamefont
  {Glauber}},\ }\href {\doibase 10.1103/PhysRev.131.2766} {\bibfield  {journal}
  {\bibinfo  {journal} {Phys. Rev.}\ }\textbf {\bibinfo {volume} {131}},\
  \bibinfo {pages} {2766} (\bibinfo {year} {1963})}\BibitemShut {NoStop}%
\bibitem [{\citenamefont {Sudarshan}(1963)}]{PhysRevLett.10.277}%
  \BibitemOpen
  \bibfield  {author} {\bibinfo {author} {\bibfnamefont {E.~C.~G.}\
  \bibnamefont {Sudarshan}},\ }\href {\doibase 10.1103/PhysRevLett.10.277}
  {\bibfield  {journal} {\bibinfo  {journal} {Phys. Rev. Lett.}\ }\textbf
  {\bibinfo {volume} {10}},\ \bibinfo {pages} {277} (\bibinfo {year}
  {1963})}\BibitemShut {NoStop}%
\bibitem [{\citenamefont {Robinson}(1965)}]{robinson_ground_1965}%
  \BibitemOpen
  \bibfield  {author} {\bibinfo {author} {\bibfnamefont {D.~W.}\ \bibnamefont
  {Robinson}},\ }\href {\doibase 10.1007/BF01646498} {\bibfield  {journal}
  {\bibinfo  {journal} {Commun.Math. Phys.}\ }\textbf {\bibinfo {volume} {1}},\
  \bibinfo {pages} {159} (\bibinfo {year} {1965})}\BibitemShut {NoStop}%
\bibitem [{\citenamefont {Hudson}(1974)}]{HUDSON1974249}%
  \BibitemOpen
  \bibfield  {author} {\bibinfo {author} {\bibfnamefont {R.}~\bibnamefont
  {Hudson}},\ }\href {\doibase https://doi.org/10.1016/0034-4877(74)90007-X}
  {\bibfield  {journal} {\bibinfo  {journal} {Reports on Mathematical Physics}\
  }\textbf {\bibinfo {volume} {6}},\ \bibinfo {pages} {249 } (\bibinfo {year}
  {1974})}\BibitemShut {NoStop}%
\bibitem [{\citenamefont {Simon}\ \emph {et~al.}(1994)\citenamefont {Simon},
  \citenamefont {Mukunda},\ and\ \citenamefont {Dutta}}]{PhysRevA.49.1567}%
  \BibitemOpen
  \bibfield  {author} {\bibinfo {author} {\bibfnamefont {R.}~\bibnamefont
  {Simon}}, \bibinfo {author} {\bibfnamefont {N.}~\bibnamefont {Mukunda}}, \
  and\ \bibinfo {author} {\bibfnamefont {B.}~\bibnamefont {Dutta}},\ }\href
  {\doibase 10.1103/PhysRevA.49.1567} {\bibfield  {journal} {\bibinfo
  {journal} {Phys. Rev. A}\ }\textbf {\bibinfo {volume} {49}},\ \bibinfo
  {pages} {1567} (\bibinfo {year} {1994})}\BibitemShut {NoStop}%
\bibitem [{\citenamefont {Braunstein}\ and\ \citenamefont {van
  Loock}(2005)}]{RevModPhys.77.513}%
  \BibitemOpen
  \bibfield  {author} {\bibinfo {author} {\bibfnamefont {S.~L.}\ \bibnamefont
  {Braunstein}}\ and\ \bibinfo {author} {\bibfnamefont {P.}~\bibnamefont {van
  Loock}},\ }\href {\doibase 10.1103/RevModPhys.77.513} {\bibfield  {journal}
  {\bibinfo  {journal} {Rev. Mod. Phys.}\ }\textbf {\bibinfo {volume} {77}},\
  \bibinfo {pages} {513} (\bibinfo {year} {2005})}\BibitemShut {NoStop}%
\bibitem [{\citenamefont {Weedbrook}\ \emph {et~al.}(2012)\citenamefont
  {Weedbrook}, \citenamefont {Pirandola}, \citenamefont {Garc\'{\i}a-Patr\'on},
  \citenamefont {Cerf}, \citenamefont {Ralph}, \citenamefont {Shapiro},\ and\
  \citenamefont {Lloyd}}]{RevModPhys.84.621}%
  \BibitemOpen
  \bibfield  {author} {\bibinfo {author} {\bibfnamefont {C.}~\bibnamefont
  {Weedbrook}}, \bibinfo {author} {\bibfnamefont {S.}~\bibnamefont
  {Pirandola}}, \bibinfo {author} {\bibfnamefont {R.}~\bibnamefont
  {Garc\'{\i}a-Patr\'on}}, \bibinfo {author} {\bibfnamefont {N.~J.}\
  \bibnamefont {Cerf}}, \bibinfo {author} {\bibfnamefont {T.~C.}\ \bibnamefont
  {Ralph}}, \bibinfo {author} {\bibfnamefont {J.~H.}\ \bibnamefont {Shapiro}},
  \ and\ \bibinfo {author} {\bibfnamefont {S.}~\bibnamefont {Lloyd}},\ }\href
  {\doibase 10.1103/RevModPhys.84.621} {\bibfield  {journal} {\bibinfo
  {journal} {Rev. Mod. Phys.}\ }\textbf {\bibinfo {volume} {84}},\ \bibinfo
  {pages} {621} (\bibinfo {year} {2012})}\BibitemShut {NoStop}%
\bibitem [{\citenamefont {Adesso}\ \emph {et~al.}(2014)\citenamefont {Adesso},
  \citenamefont {Ragy},\ and\ \citenamefont
  {Lee}}]{doi:10.1142/S1230161214400010}%
  \BibitemOpen
  \bibfield  {author} {\bibinfo {author} {\bibfnamefont {G.}~\bibnamefont
  {Adesso}}, \bibinfo {author} {\bibfnamefont {S.}~\bibnamefont {Ragy}}, \ and\
  \bibinfo {author} {\bibfnamefont {A.~R.}\ \bibnamefont {Lee}},\ }\href
  {\doibase 10.1142/S1230161214400010} {\bibfield  {journal} {\bibinfo
  {journal} {Open Systems \& Information Dynamics}\ }\textbf {\bibinfo {volume}
  {21}},\ \bibinfo {pages} {1440001} (\bibinfo {year} {2014})}\BibitemShut
  {NoStop}%
\bibitem [{\citenamefont {Verbeure}(2011)}]{verbeure_many-body_2011}%
  \BibitemOpen
  \bibfield  {author} {\bibinfo {author} {\bibfnamefont {A.}~\bibnamefont
  {Verbeure}},\ }\href@noop {} {\emph {\bibinfo {title} {Many-body boson
  systems: half a century later}}},\ Theoretical and mathematical physics\
  (\bibinfo  {publisher} {Springer},\ \bibinfo {address} {London ; New York},\
  \bibinfo {year} {2011})\BibitemShut {NoStop}%
\bibitem [{\citenamefont {Yuen}(1976)}]{PhysRevA.13.2226}%
  \BibitemOpen
  \bibfield  {author} {\bibinfo {author} {\bibfnamefont {H.~P.}\ \bibnamefont
  {Yuen}},\ }\href {\doibase 10.1103/PhysRevA.13.2226} {\bibfield  {journal}
  {\bibinfo  {journal} {Phys. Rev. A}\ }\textbf {\bibinfo {volume} {13}},\
  \bibinfo {pages} {2226} (\bibinfo {year} {1976})}\BibitemShut {NoStop}%
\bibitem [{\citenamefont {Slusher}\ \emph {et~al.}(1985)\citenamefont
  {Slusher}, \citenamefont {Hollberg}, \citenamefont {Yurke}, \citenamefont
  {Mertz},\ and\ \citenamefont {Valley}}]{PhysRevLett.55.2409}%
  \BibitemOpen
  \bibfield  {author} {\bibinfo {author} {\bibfnamefont {R.~E.}\ \bibnamefont
  {Slusher}}, \bibinfo {author} {\bibfnamefont {L.~W.}\ \bibnamefont
  {Hollberg}}, \bibinfo {author} {\bibfnamefont {B.}~\bibnamefont {Yurke}},
  \bibinfo {author} {\bibfnamefont {J.~C.}\ \bibnamefont {Mertz}}, \ and\
  \bibinfo {author} {\bibfnamefont {J.~F.}\ \bibnamefont {Valley}},\ }\href
  {\doibase 10.1103/PhysRevLett.55.2409} {\bibfield  {journal} {\bibinfo
  {journal} {Phys. Rev. Lett.}\ }\textbf {\bibinfo {volume} {55}},\ \bibinfo
  {pages} {2409} (\bibinfo {year} {1985})}\BibitemShut {NoStop}%
\bibitem [{\citenamefont {Wu}\ \emph {et~al.}(1986)\citenamefont {Wu},
  \citenamefont {Kimble}, \citenamefont {Hall},\ and\ \citenamefont
  {Wu}}]{PhysRevLett.57.2520}%
  \BibitemOpen
  \bibfield  {author} {\bibinfo {author} {\bibfnamefont {L.-A.}\ \bibnamefont
  {Wu}}, \bibinfo {author} {\bibfnamefont {H.~J.}\ \bibnamefont {Kimble}},
  \bibinfo {author} {\bibfnamefont {J.~L.}\ \bibnamefont {Hall}}, \ and\
  \bibinfo {author} {\bibfnamefont {H.}~\bibnamefont {Wu}},\ }\href {\doibase
  10.1103/PhysRevLett.57.2520} {\bibfield  {journal} {\bibinfo  {journal}
  {Phys. Rev. Lett.}\ }\textbf {\bibinfo {volume} {57}},\ \bibinfo {pages}
  {2520} (\bibinfo {year} {1986})}\BibitemShut {NoStop}%
\bibitem [{\citenamefont {Shelby}\ \emph {et~al.}(1986)\citenamefont {Shelby},
  \citenamefont {Levenson}, \citenamefont {Perlmutter}, \citenamefont {DeVoe},\
  and\ \citenamefont {Walls}}]{PhysRevLett.57.691}%
  \BibitemOpen
  \bibfield  {author} {\bibinfo {author} {\bibfnamefont {R.~M.}\ \bibnamefont
  {Shelby}}, \bibinfo {author} {\bibfnamefont {M.~D.}\ \bibnamefont
  {Levenson}}, \bibinfo {author} {\bibfnamefont {S.~H.}\ \bibnamefont
  {Perlmutter}}, \bibinfo {author} {\bibfnamefont {R.~G.}\ \bibnamefont
  {DeVoe}}, \ and\ \bibinfo {author} {\bibfnamefont {D.~F.}\ \bibnamefont
  {Walls}},\ }\href {\doibase 10.1103/PhysRevLett.57.691} {\bibfield  {journal}
  {\bibinfo  {journal} {Phys. Rev. Lett.}\ }\textbf {\bibinfo {volume} {57}},\
  \bibinfo {pages} {691} (\bibinfo {year} {1986})}\BibitemShut {NoStop}%
\bibitem [{\citenamefont {Heidmann}\ \emph {et~al.}(1987)\citenamefont
  {Heidmann}, \citenamefont {Horowicz}, \citenamefont {Reynaud}, \citenamefont
  {Giacobino}, \citenamefont {Fabre},\ and\ \citenamefont
  {Camy}}]{PhysRevLett.59.2555}%
  \BibitemOpen
  \bibfield  {author} {\bibinfo {author} {\bibfnamefont {A.}~\bibnamefont
  {Heidmann}}, \bibinfo {author} {\bibfnamefont {R.~J.}\ \bibnamefont
  {Horowicz}}, \bibinfo {author} {\bibfnamefont {S.}~\bibnamefont {Reynaud}},
  \bibinfo {author} {\bibfnamefont {E.}~\bibnamefont {Giacobino}}, \bibinfo
  {author} {\bibfnamefont {C.}~\bibnamefont {Fabre}}, \ and\ \bibinfo {author}
  {\bibfnamefont {G.}~\bibnamefont {Camy}},\ }\href {\doibase
  10.1103/PhysRevLett.59.2555} {\bibfield  {journal} {\bibinfo  {journal}
  {Phys. Rev. Lett.}\ }\textbf {\bibinfo {volume} {59}},\ \bibinfo {pages}
  {2555} (\bibinfo {year} {1987})}\BibitemShut {NoStop}%
\bibitem [{\citenamefont {Vahlbruch}\ \emph {et~al.}(2016)\citenamefont
  {Vahlbruch}, \citenamefont {Mehmet}, \citenamefont {Danzmann},\ and\
  \citenamefont {Schnabel}}]{PhysRevLett.117.110801}%
  \BibitemOpen
  \bibfield  {author} {\bibinfo {author} {\bibfnamefont {H.}~\bibnamefont
  {Vahlbruch}}, \bibinfo {author} {\bibfnamefont {M.}~\bibnamefont {Mehmet}},
  \bibinfo {author} {\bibfnamefont {K.}~\bibnamefont {Danzmann}}, \ and\
  \bibinfo {author} {\bibfnamefont {R.}~\bibnamefont {Schnabel}},\ }\href
  {\doibase 10.1103/PhysRevLett.117.110801} {\bibfield  {journal} {\bibinfo
  {journal} {Phys. Rev. Lett.}\ }\textbf {\bibinfo {volume} {117}},\ \bibinfo
  {pages} {110801} (\bibinfo {year} {2016})}\BibitemShut {NoStop}%
\bibitem [{\citenamefont {Caves}(1981)}]{PhysRevD.23.1693}%
  \BibitemOpen
  \bibfield  {author} {\bibinfo {author} {\bibfnamefont {C.~M.}\ \bibnamefont
  {Caves}},\ }\href {\doibase 10.1103/PhysRevD.23.1693} {\bibfield  {journal}
  {\bibinfo  {journal} {Phys. Rev. D}\ }\textbf {\bibinfo {volume} {23}},\
  \bibinfo {pages} {1693} (\bibinfo {year} {1981})}\BibitemShut {NoStop}%
\bibitem [{\citenamefont {Treps}\ \emph {et~al.}(2003)\citenamefont {Treps},
  \citenamefont {Grosse}, \citenamefont {Bowen}, \citenamefont {Fabre},
  \citenamefont {Bachor},\ and\ \citenamefont {Lam}}]{Treps940}%
  \BibitemOpen
  \bibfield  {author} {\bibinfo {author} {\bibfnamefont {N.}~\bibnamefont
  {Treps}}, \bibinfo {author} {\bibfnamefont {N.}~\bibnamefont {Grosse}},
  \bibinfo {author} {\bibfnamefont {W.~P.}\ \bibnamefont {Bowen}}, \bibinfo
  {author} {\bibfnamefont {C.}~\bibnamefont {Fabre}}, \bibinfo {author}
  {\bibfnamefont {H.-A.}\ \bibnamefont {Bachor}}, \ and\ \bibinfo {author}
  {\bibfnamefont {P.~K.}\ \bibnamefont {Lam}},\ }\href {\doibase
  10.1126/science.1086489} {\bibfield  {journal} {\bibinfo  {journal}
  {Science}\ }\textbf {\bibinfo {volume} {301}},\ \bibinfo {pages} {940}
  (\bibinfo {year} {2003})}\BibitemShut {NoStop}%
\bibitem [{\citenamefont {{Aasi, J.}}\ \emph {et~al.}(2013)\citenamefont
  {{Aasi, J.}} \emph {et~al.}}]{Ligo}%
  \BibitemOpen
  \bibfield  {author} {\bibinfo {author} {\bibnamefont {{Aasi, J.}}} \emph
  {et~al.},\ }\href {\doibase 10.1038/nphoton.2013.177} {\bibfield  {journal}
  {\bibinfo  {journal} {Nature Photonics}\ }\textbf {\bibinfo {volume} {7}},\
  \bibinfo {pages} {613} (\bibinfo {year} {2013})}\BibitemShut {NoStop}%
\bibitem [{\citenamefont {Schnabel}(2017)}]{SCHNABEL20171}%
  \BibitemOpen
  \bibfield  {author} {\bibinfo {author} {\bibfnamefont {R.}~\bibnamefont
  {Schnabel}},\ }\href {\doibase https://doi.org/10.1016/j.physrep.2017.04.001}
  {\bibfield  {journal} {\bibinfo  {journal} {Physics Reports}\ }\textbf
  {\bibinfo {volume} {684}},\ \bibinfo {pages} {1 } (\bibinfo {year}
  {2017})}\BibitemShut {NoStop}%
\bibitem [{\citenamefont {Duan}\ \emph {et~al.}(2000)\citenamefont {Duan},
  \citenamefont {Giedke}, \citenamefont {Cirac},\ and\ \citenamefont
  {Zoller}}]{PhysRevLett.84.2722}%
  \BibitemOpen
  \bibfield  {author} {\bibinfo {author} {\bibfnamefont {L.-M.}\ \bibnamefont
  {Duan}}, \bibinfo {author} {\bibfnamefont {G.}~\bibnamefont {Giedke}},
  \bibinfo {author} {\bibfnamefont {J.~I.}\ \bibnamefont {Cirac}}, \ and\
  \bibinfo {author} {\bibfnamefont {P.}~\bibnamefont {Zoller}},\ }\href
  {\doibase 10.1103/PhysRevLett.84.2722} {\bibfield  {journal} {\bibinfo
  {journal} {Phys. Rev. Lett.}\ }\textbf {\bibinfo {volume} {84}},\ \bibinfo
  {pages} {2722} (\bibinfo {year} {2000})}\BibitemShut {NoStop}%
\bibitem [{\citenamefont {Simon}(2000)}]{PhysRevLett.84.2726}%
  \BibitemOpen
  \bibfield  {author} {\bibinfo {author} {\bibfnamefont {R.}~\bibnamefont
  {Simon}},\ }\href {\doibase 10.1103/PhysRevLett.84.2726} {\bibfield
  {journal} {\bibinfo  {journal} {Phys. Rev. Lett.}\ }\textbf {\bibinfo
  {volume} {84}},\ \bibinfo {pages} {2726} (\bibinfo {year}
  {2000})}\BibitemShut {NoStop}%
\bibitem [{\citenamefont {Werner}\ and\ \citenamefont
  {Wolf}(2001)}]{PhysRevLett.86.3658}%
  \BibitemOpen
  \bibfield  {author} {\bibinfo {author} {\bibfnamefont {R.~F.}\ \bibnamefont
  {Werner}}\ and\ \bibinfo {author} {\bibfnamefont {M.~M.}\ \bibnamefont
  {Wolf}},\ }\href {\doibase 10.1103/PhysRevLett.86.3658} {\bibfield  {journal}
  {\bibinfo  {journal} {Phys. Rev. Lett.}\ }\textbf {\bibinfo {volume} {86}},\
  \bibinfo {pages} {3658} (\bibinfo {year} {2001})}\BibitemShut {NoStop}%
\bibitem [{\citenamefont {Braunstein}(2005)}]{PhysRevA.71.055801}%
  \BibitemOpen
  \bibfield  {author} {\bibinfo {author} {\bibfnamefont {S.~L.}\ \bibnamefont
  {Braunstein}},\ }\href {\doibase 10.1103/PhysRevA.71.055801} {\bibfield
  {journal} {\bibinfo  {journal} {Phys. Rev. A}\ }\textbf {\bibinfo {volume}
  {71}},\ \bibinfo {pages} {055801} (\bibinfo {year} {2005})}\BibitemShut
  {NoStop}%
\bibitem [{\citenamefont {Adesso}\ and\ \citenamefont
  {Illuminati}(2007)}]{Adesso_2007}%
  \BibitemOpen
  \bibfield  {author} {\bibinfo {author} {\bibfnamefont {G.}~\bibnamefont
  {Adesso}}\ and\ \bibinfo {author} {\bibfnamefont {F.}~\bibnamefont
  {Illuminati}},\ }\href {\doibase 10.1088/1751-8113/40/28/s01} {\bibfield
  {journal} {\bibinfo  {journal} {Journal of Physics A: Mathematical and
  Theoretical}\ }\textbf {\bibinfo {volume} {40}},\ \bibinfo {pages} {7821}
  (\bibinfo {year} {2007})}\BibitemShut {NoStop}%
\bibitem [{\citenamefont {Gerke}\ \emph {et~al.}(2015)\citenamefont {Gerke},
  \citenamefont {Sperling}, \citenamefont {Vogel}, \citenamefont {Cai},
  \citenamefont {Roslund}, \citenamefont {Treps},\ and\ \citenamefont
  {Fabre}}]{gerke_full_2015}%
  \BibitemOpen
  \bibfield  {author} {\bibinfo {author} {\bibfnamefont {S.}~\bibnamefont
  {Gerke}}, \bibinfo {author} {\bibfnamefont {J.}~\bibnamefont {Sperling}},
  \bibinfo {author} {\bibfnamefont {W.}~\bibnamefont {Vogel}}, \bibinfo
  {author} {\bibfnamefont {Y.}~\bibnamefont {Cai}}, \bibinfo {author}
  {\bibfnamefont {J.}~\bibnamefont {Roslund}}, \bibinfo {author} {\bibfnamefont
  {N.}~\bibnamefont {Treps}}, \ and\ \bibinfo {author} {\bibfnamefont
  {C.}~\bibnamefont {Fabre}},\ }\href {\doibase 10.1103/PhysRevLett.114.050501}
  {\bibfield  {journal} {\bibinfo  {journal} {Phys. Rev. Lett.}\ }\textbf
  {\bibinfo {volume} {114}},\ \bibinfo {pages} {050501} (\bibinfo {year}
  {2015})}\BibitemShut {NoStop}%
\bibitem [{\citenamefont {Uola}\ \emph {et~al.}(2020)\citenamefont {Uola},
  \citenamefont {Costa}, \citenamefont {Nguyen},\ and\ \citenamefont
  {G\"uhne}}]{RevModPhys.92.015001}%
  \BibitemOpen
  \bibfield  {author} {\bibinfo {author} {\bibfnamefont {R.}~\bibnamefont
  {Uola}}, \bibinfo {author} {\bibfnamefont {A.~C.~S.}\ \bibnamefont {Costa}},
  \bibinfo {author} {\bibfnamefont {H.~C.}\ \bibnamefont {Nguyen}}, \ and\
  \bibinfo {author} {\bibfnamefont {O.}~\bibnamefont {G\"uhne}},\ }\href
  {\doibase 10.1103/RevModPhys.92.015001} {\bibfield  {journal} {\bibinfo
  {journal} {Rev. Mod. Phys.}\ }\textbf {\bibinfo {volume} {92}},\ \bibinfo
  {pages} {015001} (\bibinfo {year} {2020})}\BibitemShut {NoStop}%
\bibitem [{\citenamefont {Wiseman}\ \emph {et~al.}(2007)\citenamefont
  {Wiseman}, \citenamefont {Jones},\ and\ \citenamefont
  {Doherty}}]{Wiseman:2007aa}%
  \BibitemOpen
  \bibfield  {author} {\bibinfo {author} {\bibfnamefont {H.~M.}\ \bibnamefont
  {Wiseman}}, \bibinfo {author} {\bibfnamefont {S.~J.}\ \bibnamefont {Jones}},
  \ and\ \bibinfo {author} {\bibfnamefont {A.~C.}\ \bibnamefont {Doherty}},\
  }\href {\doibase 10.1103/PhysRevLett.98.140402} {\bibfield  {journal}
  {\bibinfo  {journal} {Physical Review Letters}\ }\textbf {\bibinfo {volume}
  {98}},\ \bibinfo {pages} {140402} (\bibinfo {year} {2007})}\BibitemShut
  {NoStop}%
\bibitem [{\citenamefont {Kogias}\ \emph {et~al.}(2015)\citenamefont {Kogias},
  \citenamefont {Lee}, \citenamefont {Ragy},\ and\ \citenamefont
  {Adesso}}]{Kogias:2015aa}%
  \BibitemOpen
  \bibfield  {author} {\bibinfo {author} {\bibfnamefont {I.}~\bibnamefont
  {Kogias}}, \bibinfo {author} {\bibfnamefont {A.~R.}\ \bibnamefont {Lee}},
  \bibinfo {author} {\bibfnamefont {S.}~\bibnamefont {Ragy}}, \ and\ \bibinfo
  {author} {\bibfnamefont {G.}~\bibnamefont {Adesso}},\ }\href {\doibase
  10.1103/PhysRevLett.114.060403} {\bibfield  {journal} {\bibinfo  {journal}
  {Physical Review Letters}\ }\textbf {\bibinfo {volume} {114}},\ \bibinfo
  {pages} {060403} (\bibinfo {year} {2015})}\BibitemShut {NoStop}%
\bibitem [{\citenamefont {Deng}\ \emph {et~al.}(2017)\citenamefont {Deng},
  \citenamefont {Xiang}, \citenamefont {Tian}, \citenamefont {Adesso},
  \citenamefont {He}, \citenamefont {Gong}, \citenamefont {Su}, \citenamefont
  {Xie},\ and\ \citenamefont {Peng}}]{Deng:2017aa}%
  \BibitemOpen
  \bibfield  {author} {\bibinfo {author} {\bibfnamefont {X.}~\bibnamefont
  {Deng}}, \bibinfo {author} {\bibfnamefont {Y.}~\bibnamefont {Xiang}},
  \bibinfo {author} {\bibfnamefont {C.}~\bibnamefont {Tian}}, \bibinfo {author}
  {\bibfnamefont {G.}~\bibnamefont {Adesso}}, \bibinfo {author} {\bibfnamefont
  {Q.}~\bibnamefont {He}}, \bibinfo {author} {\bibfnamefont {Q.}~\bibnamefont
  {Gong}}, \bibinfo {author} {\bibfnamefont {X.}~\bibnamefont {Su}}, \bibinfo
  {author} {\bibfnamefont {C.}~\bibnamefont {Xie}}, \ and\ \bibinfo {author}
  {\bibfnamefont {K.}~\bibnamefont {Peng}},\ }\href {\doibase
  10.1103/PhysRevLett.118.230501} {\bibfield  {journal} {\bibinfo  {journal}
  {Physical Review Letters}\ }\textbf {\bibinfo {volume} {118}},\ \bibinfo
  {pages} {230501} (\bibinfo {year} {2017})}\BibitemShut {NoStop}%
\bibitem [{\citenamefont {Cai}\ \emph {et~al.}(2019)\citenamefont {Cai},
  \citenamefont {Xiang}, \citenamefont {Liu}, \citenamefont {He},\ and\
  \citenamefont {Treps}}]{Cai-Steering}%
  \BibitemOpen
  \bibfield  {author} {\bibinfo {author} {\bibfnamefont {Y.}~\bibnamefont
  {Cai}}, \bibinfo {author} {\bibfnamefont {Y.}~\bibnamefont {Xiang}}, \bibinfo
  {author} {\bibfnamefont {Y.}~\bibnamefont {Liu}}, \bibinfo {author}
  {\bibfnamefont {Q.}~\bibnamefont {He}}, \ and\ \bibinfo {author}
  {\bibfnamefont {N.}~\bibnamefont {Treps}},\ }\href
  {https://arxiv.org/abs/1910.13698} {\bibfield  {journal} {\bibinfo  {journal}
  {arXiv:1910.13698}\ } (\bibinfo {year} {2019})}\BibitemShut {NoStop}%
\bibitem [{\citenamefont {Bartlett}\ \emph {et~al.}(2002)\citenamefont
  {Bartlett}, \citenamefont {Sanders}, \citenamefont {Braunstein},\ and\
  \citenamefont {Nemoto}}]{PhysRevLett.88.097904}%
  \BibitemOpen
  \bibfield  {author} {\bibinfo {author} {\bibfnamefont {S.~D.}\ \bibnamefont
  {Bartlett}}, \bibinfo {author} {\bibfnamefont {B.~C.}\ \bibnamefont
  {Sanders}}, \bibinfo {author} {\bibfnamefont {S.~L.}\ \bibnamefont
  {Braunstein}}, \ and\ \bibinfo {author} {\bibfnamefont {K.}~\bibnamefont
  {Nemoto}},\ }\href {\doibase 10.1103/PhysRevLett.88.097904} {\bibfield
  {journal} {\bibinfo  {journal} {Phys. Rev. Lett.}\ }\textbf {\bibinfo
  {volume} {88}},\ \bibinfo {pages} {097904} (\bibinfo {year}
  {2002})}\BibitemShut {NoStop}%
\bibitem [{\citenamefont {Mari}\ and\ \citenamefont
  {Eisert}(2012)}]{mari_positive_2012}%
  \BibitemOpen
  \bibfield  {author} {\bibinfo {author} {\bibfnamefont {A.}~\bibnamefont
  {Mari}}\ and\ \bibinfo {author} {\bibfnamefont {J.}~\bibnamefont {Eisert}},\
  }\href {\doibase 10.1103/PhysRevLett.109.230503} {\bibfield  {journal}
  {\bibinfo  {journal} {Phys. Rev. Lett.}\ }\textbf {\bibinfo {volume} {109}},\
  \bibinfo {pages} {230503} (\bibinfo {year} {2012})}\BibitemShut {NoStop}%
\bibitem [{\citenamefont {Garc{\'i}a-{\'A}lvarez}\ \emph
  {et~al.}(2020)\citenamefont {Garc{\'i}a-{\'A}lvarez}, \citenamefont
  {Calcluth}, \citenamefont {Ferraro},\ and\ \citenamefont
  {Ferrini}}]{garcalvarez2020efficient}%
  \BibitemOpen
  \bibfield  {author} {\bibinfo {author} {\bibfnamefont {L.}~\bibnamefont
  {Garc{\'i}a-{\'A}lvarez}}, \bibinfo {author} {\bibfnamefont {C.}~\bibnamefont
  {Calcluth}}, \bibinfo {author} {\bibfnamefont {A.}~\bibnamefont {Ferraro}}, \
  and\ \bibinfo {author} {\bibfnamefont {G.}~\bibnamefont {Ferrini}},\
  }\href@noop {} {\enquote {\bibinfo {title} {Efficient simulatability of
  continuous-variable circuits with large wigner negativity},}\ } (\bibinfo
  {year} {2020}),\ \Eprint {http://arxiv.org/abs/2005.12026} {arXiv:2005.12026
  [quant-ph]} \BibitemShut {NoStop}%
\bibitem [{\citenamefont {Rahimi-Keshari}\ \emph {et~al.}(2016)\citenamefont
  {Rahimi-Keshari}, \citenamefont {Ralph},\ and\ \citenamefont
  {Caves}}]{rahimi-keshari_sufficient_2016}%
  \BibitemOpen
  \bibfield  {author} {\bibinfo {author} {\bibfnamefont {S.}~\bibnamefont
  {Rahimi-Keshari}}, \bibinfo {author} {\bibfnamefont {T.~C.}\ \bibnamefont
  {Ralph}}, \ and\ \bibinfo {author} {\bibfnamefont {C.~M.}\ \bibnamefont
  {Caves}},\ }\href {\doibase 10.1103/PhysRevX.6.021039} {\bibfield  {journal}
  {\bibinfo  {journal} {Phys. Rev. X}\ }\textbf {\bibinfo {volume} {6}},\
  \bibinfo {pages} {021039} (\bibinfo {year} {2016})}\BibitemShut {NoStop}%
\bibitem [{\citenamefont {Gu}\ \emph {et~al.}(2009)\citenamefont {Gu},
  \citenamefont {Weedbrook}, \citenamefont {Menicucci}, \citenamefont {Ralph},\
  and\ \citenamefont {van Loock}}]{gu_quantum_2009}%
  \BibitemOpen
  \bibfield  {author} {\bibinfo {author} {\bibfnamefont {M.}~\bibnamefont
  {Gu}}, \bibinfo {author} {\bibfnamefont {C.}~\bibnamefont {Weedbrook}},
  \bibinfo {author} {\bibfnamefont {N.~C.}\ \bibnamefont {Menicucci}}, \bibinfo
  {author} {\bibfnamefont {T.~C.}\ \bibnamefont {Ralph}}, \ and\ \bibinfo
  {author} {\bibfnamefont {P.}~\bibnamefont {van Loock}},\ }\href {\doibase
  10.1103/PhysRevA.79.062318} {\bibfield  {journal} {\bibinfo  {journal} {Phys.
  Rev. A}\ }\textbf {\bibinfo {volume} {79}},\ \bibinfo {pages} {062318}
  (\bibinfo {year} {2009})}\BibitemShut {NoStop}%
\bibitem [{\citenamefont {van Loock}\ \emph {et~al.}(2007)\citenamefont {van
  Loock}, \citenamefont {Weedbrook},\ and\ \citenamefont
  {Gu}}]{PhysRevA.76.032321}%
  \BibitemOpen
  \bibfield  {author} {\bibinfo {author} {\bibfnamefont {P.}~\bibnamefont {van
  Loock}}, \bibinfo {author} {\bibfnamefont {C.}~\bibnamefont {Weedbrook}}, \
  and\ \bibinfo {author} {\bibfnamefont {M.}~\bibnamefont {Gu}},\ }\href
  {\doibase 10.1103/PhysRevA.76.032321} {\bibfield  {journal} {\bibinfo
  {journal} {Phys. Rev. A}\ }\textbf {\bibinfo {volume} {76}},\ \bibinfo
  {pages} {032321} (\bibinfo {year} {2007})}\BibitemShut {NoStop}%
\bibitem [{\citenamefont {Su}\ \emph {et~al.}(2012)\citenamefont {Su},
  \citenamefont {Zhao}, \citenamefont {Hao}, \citenamefont {Jia}, \citenamefont
  {Xie},\ and\ \citenamefont {Peng}}]{Su:12}%
  \BibitemOpen
  \bibfield  {author} {\bibinfo {author} {\bibfnamefont {X.}~\bibnamefont
  {Su}}, \bibinfo {author} {\bibfnamefont {Y.}~\bibnamefont {Zhao}}, \bibinfo
  {author} {\bibfnamefont {S.}~\bibnamefont {Hao}}, \bibinfo {author}
  {\bibfnamefont {X.}~\bibnamefont {Jia}}, \bibinfo {author} {\bibfnamefont
  {C.}~\bibnamefont {Xie}}, \ and\ \bibinfo {author} {\bibfnamefont
  {K.}~\bibnamefont {Peng}},\ }\href {\doibase 10.1364/OL.37.005178} {\bibfield
   {journal} {\bibinfo  {journal} {Opt. Lett.}\ }\textbf {\bibinfo {volume}
  {37}},\ \bibinfo {pages} {5178} (\bibinfo {year} {2012})}\BibitemShut
  {NoStop}%
\bibitem [{\citenamefont {Chen}\ \emph {et~al.}(2014)\citenamefont {Chen},
  \citenamefont {Menicucci},\ and\ \citenamefont
  {Pfister}}]{PhysRevLett.112.120505}%
  \BibitemOpen
  \bibfield  {author} {\bibinfo {author} {\bibfnamefont {M.}~\bibnamefont
  {Chen}}, \bibinfo {author} {\bibfnamefont {N.~C.}\ \bibnamefont {Menicucci}},
  \ and\ \bibinfo {author} {\bibfnamefont {O.}~\bibnamefont {Pfister}},\ }\href
  {\doibase 10.1103/PhysRevLett.112.120505} {\bibfield  {journal} {\bibinfo
  {journal} {Phys. Rev. Lett.}\ }\textbf {\bibinfo {volume} {112}},\ \bibinfo
  {pages} {120505} (\bibinfo {year} {2014})}\BibitemShut {NoStop}%
\bibitem [{\citenamefont {Cai}\ \emph {et~al.}(2017)\citenamefont {Cai},
  \citenamefont {Roslund}, \citenamefont {Ferrini}, \citenamefont {Arzani},
  \citenamefont {Xu}, \citenamefont {Fabre},\ and\ \citenamefont
  {Treps}}]{cai-2017}%
  \BibitemOpen
  \bibfield  {author} {\bibinfo {author} {\bibfnamefont {Y.}~\bibnamefont
  {Cai}}, \bibinfo {author} {\bibfnamefont {J.}~\bibnamefont {Roslund}},
  \bibinfo {author} {\bibfnamefont {G.}~\bibnamefont {Ferrini}}, \bibinfo
  {author} {\bibfnamefont {F.}~\bibnamefont {Arzani}}, \bibinfo {author}
  {\bibfnamefont {X.}~\bibnamefont {Xu}}, \bibinfo {author} {\bibfnamefont
  {C.}~\bibnamefont {Fabre}}, \ and\ \bibinfo {author} {\bibfnamefont
  {N.}~\bibnamefont {Treps}},\ }\href {http://dx.doi.org/10.1038/ncomms15645}
  {\bibfield  {journal} {\bibinfo  {journal} {Nat. Commun.}\ }\textbf {\bibinfo
  {volume} {8}},\ \bibinfo {pages} {15645} (\bibinfo {year}
  {2017})}\BibitemShut {NoStop}%
\bibitem [{\citenamefont {Asavanant}\ \emph {et~al.}(2019)\citenamefont
  {Asavanant}, \citenamefont {Shiozawa}, \citenamefont {Yokoyama},
  \citenamefont {Charoensombutamon}, \citenamefont {Emura}, \citenamefont
  {Alexander}, \citenamefont {Takeda}, \citenamefont {Yoshikawa}, \citenamefont
  {Menicucci}, \citenamefont {Yonezawa},\ and\ \citenamefont
  {Furusawa}}]{Asavanant:2019aa}%
  \BibitemOpen
  \bibfield  {author} {\bibinfo {author} {\bibfnamefont {W.}~\bibnamefont
  {Asavanant}}, \bibinfo {author} {\bibfnamefont {Y.}~\bibnamefont {Shiozawa}},
  \bibinfo {author} {\bibfnamefont {S.}~\bibnamefont {Yokoyama}}, \bibinfo
  {author} {\bibfnamefont {B.}~\bibnamefont {Charoensombutamon}}, \bibinfo
  {author} {\bibfnamefont {H.}~\bibnamefont {Emura}}, \bibinfo {author}
  {\bibfnamefont {R.~N.}\ \bibnamefont {Alexander}}, \bibinfo {author}
  {\bibfnamefont {S.}~\bibnamefont {Takeda}}, \bibinfo {author} {\bibfnamefont
  {J.-i.}\ \bibnamefont {Yoshikawa}}, \bibinfo {author} {\bibfnamefont {N.~C.}\
  \bibnamefont {Menicucci}}, \bibinfo {author} {\bibfnamefont {H.}~\bibnamefont
  {Yonezawa}}, \ and\ \bibinfo {author} {\bibfnamefont {A.}~\bibnamefont
  {Furusawa}},\ }\href {\doibase 10.1126/science.aay2645} {\bibfield  {journal}
  {\bibinfo  {journal} {Science}\ }\textbf {\bibinfo {volume} {366}},\ \bibinfo
  {pages} {373} (\bibinfo {year} {2019})}\BibitemShut {NoStop}%
\bibitem [{\citenamefont {Larsen}\ \emph {et~al.}(2019)\citenamefont {Larsen},
  \citenamefont {Guo}, \citenamefont {Breum}, \citenamefont
  {Neergaard-Nielsen},\ and\ \citenamefont {Andersen}}]{Larsen:2019aa}%
  \BibitemOpen
  \bibfield  {author} {\bibinfo {author} {\bibfnamefont {M.~V.}\ \bibnamefont
  {Larsen}}, \bibinfo {author} {\bibfnamefont {X.}~\bibnamefont {Guo}},
  \bibinfo {author} {\bibfnamefont {C.~R.}\ \bibnamefont {Breum}}, \bibinfo
  {author} {\bibfnamefont {J.~S.}\ \bibnamefont {Neergaard-Nielsen}}, \ and\
  \bibinfo {author} {\bibfnamefont {U.~L.}\ \bibnamefont {Andersen}},\ }\href
  {\doibase 10.1126/science.aay4354} {\bibfield  {journal} {\bibinfo  {journal}
  {Science}\ }\textbf {\bibinfo {volume} {366}},\ \bibinfo {pages} {369}
  (\bibinfo {year} {2019})}\BibitemShut {NoStop}%
\bibitem [{\citenamefont {Fiur\'a\ifmmode~\check{s}\else \v{s}\fi{}ek}\ \emph
  {et~al.}(2005)\citenamefont {Fiur\'a\ifmmode~\check{s}\else \v{s}\fi{}ek},
  \citenamefont {Garc\'{\i}a-Patr\'on},\ and\ \citenamefont
  {Cerf}}]{PhysRevA.72.033822}%
  \BibitemOpen
  \bibfield  {author} {\bibinfo {author} {\bibfnamefont {J.}~\bibnamefont
  {Fiur\'a\ifmmode~\check{s}\else \v{s}\fi{}ek}}, \bibinfo {author}
  {\bibfnamefont {R.}~\bibnamefont {Garc\'{\i}a-Patr\'on}}, \ and\ \bibinfo
  {author} {\bibfnamefont {N.~J.}\ \bibnamefont {Cerf}},\ }\href {\doibase
  10.1103/PhysRevA.72.033822} {\bibfield  {journal} {\bibinfo  {journal} {Phys.
  Rev. A}\ }\textbf {\bibinfo {volume} {72}},\ \bibinfo {pages} {033822}
  (\bibinfo {year} {2005})}\BibitemShut {NoStop}%
\bibitem [{\citenamefont {Tualle-Brouri}\ \emph {et~al.}(2009)\citenamefont
  {Tualle-Brouri}, \citenamefont {Ourjoumtsev}, \citenamefont {Dantan},
  \citenamefont {Grangier}, \citenamefont {Wubs},\ and\ \citenamefont
  {S\o{}rensen}}]{PhysRevA.80.013806}%
  \BibitemOpen
  \bibfield  {author} {\bibinfo {author} {\bibfnamefont {R.}~\bibnamefont
  {Tualle-Brouri}}, \bibinfo {author} {\bibfnamefont {A.}~\bibnamefont
  {Ourjoumtsev}}, \bibinfo {author} {\bibfnamefont {A.}~\bibnamefont {Dantan}},
  \bibinfo {author} {\bibfnamefont {P.}~\bibnamefont {Grangier}}, \bibinfo
  {author} {\bibfnamefont {M.}~\bibnamefont {Wubs}}, \ and\ \bibinfo {author}
  {\bibfnamefont {A.~S.}\ \bibnamefont {S\o{}rensen}},\ }\href {\doibase
  10.1103/PhysRevA.80.013806} {\bibfield  {journal} {\bibinfo  {journal} {Phys.
  Rev. A}\ }\textbf {\bibinfo {volume} {80}},\ \bibinfo {pages} {013806}
  (\bibinfo {year} {2009})}\BibitemShut {NoStop}%
\bibitem [{\citenamefont {Lund}\ \emph {et~al.}(2014)\citenamefont {Lund},
  \citenamefont {Laing}, \citenamefont {Rahimi-Keshari}, \citenamefont
  {Rudolph}, \citenamefont {O'Brien},\ and\ \citenamefont
  {Ralph}}]{PhysRevLett.113.100502}%
  \BibitemOpen
  \bibfield  {author} {\bibinfo {author} {\bibfnamefont {A.~P.}\ \bibnamefont
  {Lund}}, \bibinfo {author} {\bibfnamefont {A.}~\bibnamefont {Laing}},
  \bibinfo {author} {\bibfnamefont {S.}~\bibnamefont {Rahimi-Keshari}},
  \bibinfo {author} {\bibfnamefont {T.}~\bibnamefont {Rudolph}}, \bibinfo
  {author} {\bibfnamefont {J.~L.}\ \bibnamefont {O'Brien}}, \ and\ \bibinfo
  {author} {\bibfnamefont {T.~C.}\ \bibnamefont {Ralph}},\ }\href {\doibase
  10.1103/PhysRevLett.113.100502} {\bibfield  {journal} {\bibinfo  {journal}
  {Phys. Rev. Lett.}\ }\textbf {\bibinfo {volume} {113}},\ \bibinfo {pages}
  {100502} (\bibinfo {year} {2014})}\BibitemShut {NoStop}%
\bibitem [{\citenamefont {Su}\ \emph {et~al.}(2019)\citenamefont {Su},
  \citenamefont {Myers},\ and\ \citenamefont {Sabapathy}}]{su2019generation}%
  \BibitemOpen
  \bibfield  {author} {\bibinfo {author} {\bibfnamefont {D.}~\bibnamefont
  {Su}}, \bibinfo {author} {\bibfnamefont {C.~R.}\ \bibnamefont {Myers}}, \
  and\ \bibinfo {author} {\bibfnamefont {K.~K.}\ \bibnamefont {Sabapathy}},\
  }\href {\doibase 10.1103/PhysRevA.100.052301} {\bibfield  {journal} {\bibinfo
   {journal} {Phys. Rev. A}\ }\textbf {\bibinfo {volume} {100}},\ \bibinfo
  {pages} {052301} (\bibinfo {year} {2019})}\BibitemShut {NoStop}%
\bibitem [{\citenamefont {Lvovsky}\ \emph {et~al.}(2020)\citenamefont
  {Lvovsky}, \citenamefont {Grangier}, \citenamefont {Ourjoumtsev},
  \citenamefont {Parigi}, \citenamefont {Sasaki},\ and\ \citenamefont
  {Tualle-Brouri}}]{lvovsky2020production}%
  \BibitemOpen
  \bibfield  {author} {\bibinfo {author} {\bibfnamefont {A.~I.}\ \bibnamefont
  {Lvovsky}}, \bibinfo {author} {\bibfnamefont {P.}~\bibnamefont {Grangier}},
  \bibinfo {author} {\bibfnamefont {A.}~\bibnamefont {Ourjoumtsev}}, \bibinfo
  {author} {\bibfnamefont {V.}~\bibnamefont {Parigi}}, \bibinfo {author}
  {\bibfnamefont {M.}~\bibnamefont {Sasaki}}, \ and\ \bibinfo {author}
  {\bibfnamefont {R.}~\bibnamefont {Tualle-Brouri}},\ }\href@noop {} {\enquote
  {\bibinfo {title} {Production and applications of non-gaussian quantum states
  of light},}\ } (\bibinfo {year} {2020}),\ \Eprint
  {http://arxiv.org/abs/2006.16985} {arXiv:2006.16985 [quant-ph]} \BibitemShut
  {NoStop}%
\bibitem [{\citenamefont {Hong}\ and\ \citenamefont
  {Mandel}(1986)}]{PhysRevLett.56.58}%
  \BibitemOpen
  \bibfield  {author} {\bibinfo {author} {\bibfnamefont {C.~K.}\ \bibnamefont
  {Hong}}\ and\ \bibinfo {author} {\bibfnamefont {L.}~\bibnamefont {Mandel}},\
  }\href {\doibase 10.1103/PhysRevLett.56.58} {\bibfield  {journal} {\bibinfo
  {journal} {Phys. Rev. Lett.}\ }\textbf {\bibinfo {volume} {56}},\ \bibinfo
  {pages} {58} (\bibinfo {year} {1986})}\BibitemShut {NoStop}%
\bibitem [{\citenamefont {Lvovsky}\ \emph {et~al.}(2001)\citenamefont
  {Lvovsky}, \citenamefont {Hansen}, \citenamefont {Aichele}, \citenamefont
  {Benson}, \citenamefont {Mlynek},\ and\ \citenamefont
  {Schiller}}]{PhysRevLett.87.050402}%
  \BibitemOpen
  \bibfield  {author} {\bibinfo {author} {\bibfnamefont {A.~I.}\ \bibnamefont
  {Lvovsky}}, \bibinfo {author} {\bibfnamefont {H.}~\bibnamefont {Hansen}},
  \bibinfo {author} {\bibfnamefont {T.}~\bibnamefont {Aichele}}, \bibinfo
  {author} {\bibfnamefont {O.}~\bibnamefont {Benson}}, \bibinfo {author}
  {\bibfnamefont {J.}~\bibnamefont {Mlynek}}, \ and\ \bibinfo {author}
  {\bibfnamefont {S.}~\bibnamefont {Schiller}},\ }\href {\doibase
  10.1103/PhysRevLett.87.050402} {\bibfield  {journal} {\bibinfo  {journal}
  {Phys. Rev. Lett.}\ }\textbf {\bibinfo {volume} {87}},\ \bibinfo {pages}
  {050402} (\bibinfo {year} {2001})}\BibitemShut {NoStop}%
\bibitem [{\citenamefont {Bra{\'{n}}czyk}\ \emph {et~al.}(2010)\citenamefont
  {Bra{\'{n}}czyk}, \citenamefont {Ralph}, \citenamefont {Helwig},\ and\
  \citenamefont {Silberhorn}}]{Bra_czyk_2010}%
  \BibitemOpen
  \bibfield  {author} {\bibinfo {author} {\bibfnamefont {A.~M.}\ \bibnamefont
  {Bra{\'{n}}czyk}}, \bibinfo {author} {\bibfnamefont {T.~C.}\ \bibnamefont
  {Ralph}}, \bibinfo {author} {\bibfnamefont {W.}~\bibnamefont {Helwig}}, \
  and\ \bibinfo {author} {\bibfnamefont {C.}~\bibnamefont {Silberhorn}},\
  }\href {\doibase 10.1088/1367-2630/12/6/063001} {\bibfield  {journal}
  {\bibinfo  {journal} {New Journal of Physics}\ }\textbf {\bibinfo {volume}
  {12}},\ \bibinfo {pages} {063001} (\bibinfo {year} {2010})}\BibitemShut
  {NoStop}%
\bibitem [{\citenamefont {Wenger}\ \emph {et~al.}(2004)\citenamefont {Wenger},
  \citenamefont {Tualle-Brouri},\ and\ \citenamefont
  {Grangier}}]{Wenger:2004aa}%
  \BibitemOpen
  \bibfield  {author} {\bibinfo {author} {\bibfnamefont {J.}~\bibnamefont
  {Wenger}}, \bibinfo {author} {\bibfnamefont {R.}~\bibnamefont
  {Tualle-Brouri}}, \ and\ \bibinfo {author} {\bibfnamefont {P.}~\bibnamefont
  {Grangier}},\ }\href {\doibase 10.1103/PhysRevLett.92.153601} {\bibfield
  {journal} {\bibinfo  {journal} {Physical Review Letters}\ }\textbf {\bibinfo
  {volume} {92}},\ \bibinfo {pages} {153601} (\bibinfo {year}
  {2004})}\BibitemShut {NoStop}%
\bibitem [{\citenamefont {Ourjoumtsev}\ \emph {et~al.}(2006)\citenamefont
  {Ourjoumtsev}, \citenamefont {Tualle-Brouri}, \citenamefont {Laurat},\ and\
  \citenamefont {Grangier}}]{Ourjoumtsev83}%
  \BibitemOpen
  \bibfield  {author} {\bibinfo {author} {\bibfnamefont {A.}~\bibnamefont
  {Ourjoumtsev}}, \bibinfo {author} {\bibfnamefont {R.}~\bibnamefont
  {Tualle-Brouri}}, \bibinfo {author} {\bibfnamefont {J.}~\bibnamefont
  {Laurat}}, \ and\ \bibinfo {author} {\bibfnamefont {P.}~\bibnamefont
  {Grangier}},\ }\href {\doibase 10.1126/science.1122858} {\bibfield  {journal}
  {\bibinfo  {journal} {Science}\ }\textbf {\bibinfo {volume} {312}},\ \bibinfo
  {pages} {83} (\bibinfo {year} {2006})}\BibitemShut {NoStop}%
\bibitem [{\citenamefont {Zavatta}\ \emph {et~al.}(2004)\citenamefont
  {Zavatta}, \citenamefont {Viciani},\ and\ \citenamefont
  {Bellini}}]{Zavatta660}%
  \BibitemOpen
  \bibfield  {author} {\bibinfo {author} {\bibfnamefont {A.}~\bibnamefont
  {Zavatta}}, \bibinfo {author} {\bibfnamefont {S.}~\bibnamefont {Viciani}}, \
  and\ \bibinfo {author} {\bibfnamefont {M.}~\bibnamefont {Bellini}},\ }\href
  {\doibase 10.1126/science.1103190} {\bibfield  {journal} {\bibinfo  {journal}
  {Science}\ }\textbf {\bibinfo {volume} {306}},\ \bibinfo {pages} {660}
  (\bibinfo {year} {2004})}\BibitemShut {NoStop}%
\bibitem [{\citenamefont {Parigi}\ \emph {et~al.}(2007)\citenamefont {Parigi},
  \citenamefont {Zavatta}, \citenamefont {Kim},\ and\ \citenamefont
  {Bellini}}]{parigi_probing_2007}%
  \BibitemOpen
  \bibfield  {author} {\bibinfo {author} {\bibfnamefont {V.}~\bibnamefont
  {Parigi}}, \bibinfo {author} {\bibfnamefont {A.}~\bibnamefont {Zavatta}},
  \bibinfo {author} {\bibfnamefont {M.}~\bibnamefont {Kim}}, \ and\ \bibinfo
  {author} {\bibfnamefont {M.}~\bibnamefont {Bellini}},\ }\href {\doibase
  10.1126/science.1146204} {\bibfield  {journal} {\bibinfo  {journal}
  {Science}\ }\textbf {\bibinfo {volume} {317}},\ \bibinfo {pages} {1890}
  (\bibinfo {year} {2007})}\BibitemShut {NoStop}%
\bibitem [{\citenamefont {Averchenko}\ \emph {et~al.}(2016)\citenamefont
  {Averchenko}, \citenamefont {Jacquard}, \citenamefont {Thiel}, \citenamefont
  {Fabre},\ and\ \citenamefont {Treps}}]{averchenko_multimode_2016}%
  \BibitemOpen
  \bibfield  {author} {\bibinfo {author} {\bibfnamefont {V.~A.}\ \bibnamefont
  {Averchenko}}, \bibinfo {author} {\bibfnamefont {C.}~\bibnamefont
  {Jacquard}}, \bibinfo {author} {\bibfnamefont {V.}~\bibnamefont {Thiel}},
  \bibinfo {author} {\bibfnamefont {C.}~\bibnamefont {Fabre}}, \ and\ \bibinfo
  {author} {\bibfnamefont {N.}~\bibnamefont {Treps}},\ }\href {\doibase
  10.1088/1367-2630/18/8/083042} {\bibfield  {journal} {\bibinfo  {journal}
  {New J. Phys.}\ }\textbf {\bibinfo {volume} {18}},\ \bibinfo {pages} {083042}
  (\bibinfo {year} {2016})}\BibitemShut {NoStop}%
\bibitem [{\citenamefont {Ra}\ \emph {et~al.}(2019)\citenamefont {Ra},
  \citenamefont {Dufour}, \citenamefont {Walschaers}, \citenamefont {Jacquard},
  \citenamefont {Michel}, \citenamefont {Fabre},\ and\ \citenamefont
  {Treps}}]{Ra2019}%
  \BibitemOpen
  \bibfield  {author} {\bibinfo {author} {\bibfnamefont {Y.-S.}\ \bibnamefont
  {Ra}}, \bibinfo {author} {\bibfnamefont {A.}~\bibnamefont {Dufour}}, \bibinfo
  {author} {\bibfnamefont {M.}~\bibnamefont {Walschaers}}, \bibinfo {author}
  {\bibfnamefont {C.}~\bibnamefont {Jacquard}}, \bibinfo {author}
  {\bibfnamefont {T.}~\bibnamefont {Michel}}, \bibinfo {author} {\bibfnamefont
  {C.}~\bibnamefont {Fabre}}, \ and\ \bibinfo {author} {\bibfnamefont
  {N.}~\bibnamefont {Treps}},\ }\href {\doibase 10.1038/s41567-019-0726-y}
  {\bibfield  {journal} {\bibinfo  {journal} {Nature Physics}\ } (\bibinfo
  {year} {2019}),\ 10.1038/s41567-019-0726-y}\BibitemShut {NoStop}%
\bibitem [{\citenamefont {Dakna}\ \emph {et~al.}(1997)\citenamefont {Dakna},
  \citenamefont {Anhut}, \citenamefont {Opatrn{\'y}}, \citenamefont
  {Kn{\"o}ll},\ and\ \citenamefont {Welsch}}]{dakna_generating_1997}%
  \BibitemOpen
  \bibfield  {author} {\bibinfo {author} {\bibfnamefont {M.}~\bibnamefont
  {Dakna}}, \bibinfo {author} {\bibfnamefont {T.}~\bibnamefont {Anhut}},
  \bibinfo {author} {\bibfnamefont {T.}~\bibnamefont {Opatrn{\'y}}}, \bibinfo
  {author} {\bibfnamefont {L.}~\bibnamefont {Kn{\"o}ll}}, \ and\ \bibinfo
  {author} {\bibfnamefont {D.-G.}\ \bibnamefont {Welsch}},\ }\href {\doibase
  10.1103/PhysRevA.55.3184} {\bibfield  {journal} {\bibinfo  {journal} {Phys.
  Rev. A}\ }\textbf {\bibinfo {volume} {55}},\ \bibinfo {pages} {3184}
  (\bibinfo {year} {1997})}\BibitemShut {NoStop}%
\bibitem [{\citenamefont {Thekkadath}\ \emph {et~al.}(2020)\citenamefont
  {Thekkadath}, \citenamefont {Bell}, \citenamefont {Walmsley},\ and\
  \citenamefont {Lvovsky}}]{Thekkadath2020engineering}%
  \BibitemOpen
  \bibfield  {author} {\bibinfo {author} {\bibfnamefont {G.~S.}\ \bibnamefont
  {Thekkadath}}, \bibinfo {author} {\bibfnamefont {B.~A.}\ \bibnamefont
  {Bell}}, \bibinfo {author} {\bibfnamefont {I.~A.}\ \bibnamefont {Walmsley}},
  \ and\ \bibinfo {author} {\bibfnamefont {A.~I.}\ \bibnamefont {Lvovsky}},\
  }\href {\doibase 10.22331/q-2020-03-02-239} {\bibfield  {journal} {\bibinfo
  {journal} {{Quantum}}\ }\textbf {\bibinfo {volume} {4}},\ \bibinfo {pages}
  {239} (\bibinfo {year} {2020})}\BibitemShut {NoStop}%
\bibitem [{\citenamefont {Eaton}\ \emph {et~al.}(2019)\citenamefont {Eaton},
  \citenamefont {Nehra},\ and\ \citenamefont {Pfister}}]{Eaton_2019}%
  \BibitemOpen
  \bibfield  {author} {\bibinfo {author} {\bibfnamefont {M.}~\bibnamefont
  {Eaton}}, \bibinfo {author} {\bibfnamefont {R.}~\bibnamefont {Nehra}}, \ and\
  \bibinfo {author} {\bibfnamefont {O.}~\bibnamefont {Pfister}},\ }\href
  {\doibase 10.1088/1367-2630/ab5330} {\bibfield  {journal} {\bibinfo
  {journal} {New Journal of Physics}\ }\textbf {\bibinfo {volume} {21}},\
  \bibinfo {pages} {113034} (\bibinfo {year} {2019})}\BibitemShut {NoStop}%
\bibitem [{\citenamefont {Wigner}(1932)}]{PhysRev.40.749}%
  \BibitemOpen
  \bibfield  {author} {\bibinfo {author} {\bibfnamefont {E.}~\bibnamefont
  {Wigner}},\ }\href {\doibase 10.1103/PhysRev.40.749} {\bibfield  {journal}
  {\bibinfo  {journal} {Phys. Rev.}\ }\textbf {\bibinfo {volume} {40}},\
  \bibinfo {pages} {749} (\bibinfo {year} {1932})}\BibitemShut {NoStop}%
\bibitem [{\citenamefont {Cahill}\ and\ \citenamefont
  {Glauber}(1969)}]{PhysRev.177.1882}%
  \BibitemOpen
  \bibfield  {author} {\bibinfo {author} {\bibfnamefont {K.~E.}\ \bibnamefont
  {Cahill}}\ and\ \bibinfo {author} {\bibfnamefont {R.~J.}\ \bibnamefont
  {Glauber}},\ }\href {\doibase 10.1103/PhysRev.177.1882} {\bibfield  {journal}
  {\bibinfo  {journal} {Phys. Rev.}\ }\textbf {\bibinfo {volume} {177}},\
  \bibinfo {pages} {1882} (\bibinfo {year} {1969})}\BibitemShut {NoStop}%
\bibitem [{\citenamefont {Hillery}\ \emph {et~al.}(1984)\citenamefont
  {Hillery}, \citenamefont {O'Connell}, \citenamefont {Scully},\ and\
  \citenamefont {Wigner}}]{HILLERY1984121}%
  \BibitemOpen
  \bibfield  {author} {\bibinfo {author} {\bibfnamefont {M.}~\bibnamefont
  {Hillery}}, \bibinfo {author} {\bibfnamefont {R.}~\bibnamefont {O'Connell}},
  \bibinfo {author} {\bibfnamefont {M.}~\bibnamefont {Scully}}, \ and\ \bibinfo
  {author} {\bibfnamefont {E.}~\bibnamefont {Wigner}},\ }\href {\doibase
  https://doi.org/10.1016/0370-1573(84)90160-1} {\bibfield  {journal} {\bibinfo
   {journal} {Physics Reports}\ }\textbf {\bibinfo {volume} {106}},\ \bibinfo
  {pages} {121 } (\bibinfo {year} {1984})}\BibitemShut {NoStop}%
\bibitem [{\citenamefont {Kraus}(1971)}]{KRAUS1971311}%
  \BibitemOpen
  \bibfield  {author} {\bibinfo {author} {\bibfnamefont {K.}~\bibnamefont
  {Kraus}},\ }\href {\doibase https://doi.org/10.1016/0003-4916(71)90108-4}
  {\bibfield  {journal} {\bibinfo  {journal} {Annals of Physics}\ }\textbf
  {\bibinfo {volume} {64}},\ \bibinfo {pages} {311 } (\bibinfo {year}
  {1971})}\BibitemShut {NoStop}%
\bibitem [{Note1()}]{Note1}%
  \BibitemOpen
  \bibinfo {note} {On the full multimode Hilbert space, the operator will take
  the form $\protect \mathds {1} \otimes \protect \mathaccentV {hat}05EX_j$,
  but for simplicity we will just denote it as $\protect \mathaccentV
  {hat}05EX_j$.}\BibitemShut {Stop}%
\bibitem [{\citenamefont {Muirhead}(2008)}]{Muirhead:aa}%
  \BibitemOpen
  \bibfield  {author} {\bibinfo {author} {\bibfnamefont {R.~J.}\ \bibnamefont
  {Muirhead}},\ }\enquote {\bibinfo {title} {The multivariate normal and
  related distributions},}\ in\ \href
  {https://onlinelibrary.wiley.com/doi/abs/10.1002/9780470316559.ch1} {\emph
  {\bibinfo {booktitle} {Aspects of Multivariate Statistical Theory}}},\
  \bibinfo {series and number} {Wiley Series in Probability and Statistics}\
  (\bibinfo  {publisher} {Wiley},\ \bibinfo {year} {2008})\ pp.\ \bibinfo
  {pages} {1--49}\BibitemShut {NoStop}%
\bibitem [{\citenamefont {Walschaers}\ \emph {et~al.}(2018)\citenamefont
  {Walschaers}, \citenamefont {Sarkar}, \citenamefont {Parigi},\ and\
  \citenamefont {Treps}}]{Walschaers:2018aa}%
  \BibitemOpen
  \bibfield  {author} {\bibinfo {author} {\bibfnamefont {M.}~\bibnamefont
  {Walschaers}}, \bibinfo {author} {\bibfnamefont {S.}~\bibnamefont {Sarkar}},
  \bibinfo {author} {\bibfnamefont {V.}~\bibnamefont {Parigi}}, \ and\ \bibinfo
  {author} {\bibfnamefont {N.}~\bibnamefont {Treps}},\ }\href {\doibase
  10.1103/PhysRevLett.121.220501} {\bibfield  {journal} {\bibinfo  {journal}
  {Physical Review Letters}\ }\textbf {\bibinfo {volume} {121}},\ \bibinfo
  {pages} {220501} (\bibinfo {year} {2018})}\BibitemShut {NoStop}%
\bibitem [{\citenamefont {Cavalcanti}\ and\ \citenamefont
  {Skrzypczyk}(2017)}]{0034-4885-80-2-024001}%
  \BibitemOpen
  \bibfield  {author} {\bibinfo {author} {\bibfnamefont {D.}~\bibnamefont
  {Cavalcanti}}\ and\ \bibinfo {author} {\bibfnamefont {P.}~\bibnamefont
  {Skrzypczyk}},\ }\href {http://stacks.iop.org/0034-4885/80/i=2/a=024001}
  {\bibfield  {journal} {\bibinfo  {journal} {Reports on Progress in Physics}\
  }\textbf {\bibinfo {volume} {80}},\ \bibinfo {pages} {024001} (\bibinfo
  {year} {2017})}\BibitemShut {NoStop}%
\bibitem [{\citenamefont {Cavalcanti}\ \emph {et~al.}(2009)\citenamefont
  {Cavalcanti}, \citenamefont {Jones}, \citenamefont {Wiseman},\ and\
  \citenamefont {Reid}}]{PhysRevA.80.032112}%
  \BibitemOpen
  \bibfield  {author} {\bibinfo {author} {\bibfnamefont {E.~G.}\ \bibnamefont
  {Cavalcanti}}, \bibinfo {author} {\bibfnamefont {S.~J.}\ \bibnamefont
  {Jones}}, \bibinfo {author} {\bibfnamefont {H.~M.}\ \bibnamefont {Wiseman}},
  \ and\ \bibinfo {author} {\bibfnamefont {M.~D.}\ \bibnamefont {Reid}},\
  }\href {\doibase 10.1103/PhysRevA.80.032112} {\bibfield  {journal} {\bibinfo
  {journal} {Phys. Rev. A}\ }\textbf {\bibinfo {volume} {80}},\ \bibinfo
  {pages} {032112} (\bibinfo {year} {2009})}\BibitemShut {NoStop}%
\bibitem [{\citenamefont {Saunders}\ \emph {et~al.}(2010)\citenamefont
  {Saunders}, \citenamefont {Jones}, \citenamefont {Wiseman},\ and\
  \citenamefont {Pryde}}]{Saunders:2010aa}%
  \BibitemOpen
  \bibfield  {author} {\bibinfo {author} {\bibfnamefont {D.~J.}\ \bibnamefont
  {Saunders}}, \bibinfo {author} {\bibfnamefont {S.~J.}\ \bibnamefont {Jones}},
  \bibinfo {author} {\bibfnamefont {H.~M.}\ \bibnamefont {Wiseman}}, \ and\
  \bibinfo {author} {\bibfnamefont {G.~J.}\ \bibnamefont {Pryde}},\ }\href
  {\doibase 10.1038/nphys1766} {\bibfield  {journal} {\bibinfo  {journal}
  {Nature Physics}\ }\textbf {\bibinfo {volume} {6}},\ \bibinfo {pages} {845}
  (\bibinfo {year} {2010})}\BibitemShut {NoStop}%
\bibitem [{\citenamefont {Bennet}\ \emph {et~al.}(2012)\citenamefont {Bennet},
  \citenamefont {Evans}, \citenamefont {Saunders}, \citenamefont {Branciard},
  \citenamefont {Cavalcanti}, \citenamefont {Wiseman},\ and\ \citenamefont
  {Pryde}}]{PhysRevX.2.031003}%
  \BibitemOpen
  \bibfield  {author} {\bibinfo {author} {\bibfnamefont {A.~J.}\ \bibnamefont
  {Bennet}}, \bibinfo {author} {\bibfnamefont {D.~A.}\ \bibnamefont {Evans}},
  \bibinfo {author} {\bibfnamefont {D.~J.}\ \bibnamefont {Saunders}}, \bibinfo
  {author} {\bibfnamefont {C.}~\bibnamefont {Branciard}}, \bibinfo {author}
  {\bibfnamefont {E.~G.}\ \bibnamefont {Cavalcanti}}, \bibinfo {author}
  {\bibfnamefont {H.~M.}\ \bibnamefont {Wiseman}}, \ and\ \bibinfo {author}
  {\bibfnamefont {G.~J.}\ \bibnamefont {Pryde}},\ }\href {\doibase
  10.1103/PhysRevX.2.031003} {\bibfield  {journal} {\bibinfo  {journal} {Phys.
  Rev. X}\ }\textbf {\bibinfo {volume} {2}},\ \bibinfo {pages} {031003}
  (\bibinfo {year} {2012})}\BibitemShut {NoStop}%
\bibitem [{\citenamefont {H{\"a}ndchen}\ \emph {et~al.}(2012)\citenamefont
  {H{\"a}ndchen}, \citenamefont {Eberle}, \citenamefont {Steinlechner},
  \citenamefont {Samblowski}, \citenamefont {Franz}, \citenamefont {Werner},\
  and\ \citenamefont {Schnabel}}]{Handchen:2012aa}%
  \BibitemOpen
  \bibfield  {author} {\bibinfo {author} {\bibfnamefont {V.}~\bibnamefont
  {H{\"a}ndchen}}, \bibinfo {author} {\bibfnamefont {T.}~\bibnamefont
  {Eberle}}, \bibinfo {author} {\bibfnamefont {S.}~\bibnamefont
  {Steinlechner}}, \bibinfo {author} {\bibfnamefont {A.}~\bibnamefont
  {Samblowski}}, \bibinfo {author} {\bibfnamefont {T.}~\bibnamefont {Franz}},
  \bibinfo {author} {\bibfnamefont {R.~F.}\ \bibnamefont {Werner}}, \ and\
  \bibinfo {author} {\bibfnamefont {R.}~\bibnamefont {Schnabel}},\ }\href
  {\doibase 10.1038/nphoton.2012.202} {\bibfield  {journal} {\bibinfo
  {journal} {Nature Photonics}\ }\textbf {\bibinfo {volume} {6}},\ \bibinfo
  {pages} {596} (\bibinfo {year} {2012})}\BibitemShut {NoStop}%
\bibitem [{\citenamefont {Smith}\ \emph {et~al.}(2012)\citenamefont {Smith},
  \citenamefont {Gillett}, \citenamefont {de~Almeida}, \citenamefont
  {Branciard}, \citenamefont {Fedrizzi}, \citenamefont {Weinhold},
  \citenamefont {Lita}, \citenamefont {Calkins}, \citenamefont {Gerrits},
  \citenamefont {Wiseman}, \citenamefont {Nam},\ and\ \citenamefont
  {White}}]{10.1038/ncomms1628}%
  \BibitemOpen
  \bibfield  {author} {\bibinfo {author} {\bibfnamefont {D.~H.}\ \bibnamefont
  {Smith}}, \bibinfo {author} {\bibfnamefont {G.}~\bibnamefont {Gillett}},
  \bibinfo {author} {\bibfnamefont {M.~P.}\ \bibnamefont {de~Almeida}},
  \bibinfo {author} {\bibfnamefont {C.}~\bibnamefont {Branciard}}, \bibinfo
  {author} {\bibfnamefont {A.}~\bibnamefont {Fedrizzi}}, \bibinfo {author}
  {\bibfnamefont {T.~J.}\ \bibnamefont {Weinhold}}, \bibinfo {author}
  {\bibfnamefont {A.}~\bibnamefont {Lita}}, \bibinfo {author} {\bibfnamefont
  {B.}~\bibnamefont {Calkins}}, \bibinfo {author} {\bibfnamefont
  {T.}~\bibnamefont {Gerrits}}, \bibinfo {author} {\bibfnamefont {H.~M.}\
  \bibnamefont {Wiseman}}, \bibinfo {author} {\bibfnamefont {S.~W.}\
  \bibnamefont {Nam}}, \ and\ \bibinfo {author} {\bibfnamefont {A.~G.}\
  \bibnamefont {White}},\ }\href@noop {} {\bibfield  {journal} {\bibinfo
  {journal} {Nature Communications}\ }\textbf {\bibinfo {volume} {3}},\
  \bibinfo {pages} {625} (\bibinfo {year} {2012})}\BibitemShut {NoStop}%
\bibitem [{\citenamefont {Schneeloch}\ \emph {et~al.}(2013)\citenamefont
  {Schneeloch}, \citenamefont {Dixon}, \citenamefont {Howland}, \citenamefont
  {Broadbent},\ and\ \citenamefont {Howell}}]{PhysRevLett.110.130407}%
  \BibitemOpen
  \bibfield  {author} {\bibinfo {author} {\bibfnamefont {J.}~\bibnamefont
  {Schneeloch}}, \bibinfo {author} {\bibfnamefont {P.~B.}\ \bibnamefont
  {Dixon}}, \bibinfo {author} {\bibfnamefont {G.~A.}\ \bibnamefont {Howland}},
  \bibinfo {author} {\bibfnamefont {C.~J.}\ \bibnamefont {Broadbent}}, \ and\
  \bibinfo {author} {\bibfnamefont {J.~C.}\ \bibnamefont {Howell}},\ }\href
  {\doibase 10.1103/PhysRevLett.110.130407} {\bibfield  {journal} {\bibinfo
  {journal} {Phys. Rev. Lett.}\ }\textbf {\bibinfo {volume} {110}},\ \bibinfo
  {pages} {130407} (\bibinfo {year} {2013})}\BibitemShut {NoStop}%
\bibitem [{\citenamefont {Kocsis}\ \emph {et~al.}(2015)\citenamefont {Kocsis},
  \citenamefont {Hall}, \citenamefont {Bennet}, \citenamefont {Saunders},\ and\
  \citenamefont {Pryde}}]{10.1038/ncomms6886}%
  \BibitemOpen
  \bibfield  {author} {\bibinfo {author} {\bibfnamefont {S.}~\bibnamefont
  {Kocsis}}, \bibinfo {author} {\bibfnamefont {M.~J.~W.}\ \bibnamefont {Hall}},
  \bibinfo {author} {\bibfnamefont {A.~J.}\ \bibnamefont {Bennet}}, \bibinfo
  {author} {\bibfnamefont {D.~J.}\ \bibnamefont {Saunders}}, \ and\ \bibinfo
  {author} {\bibfnamefont {G.~J.}\ \bibnamefont {Pryde}},\ }\href {\doibase
  10.1038/ncomms6886} {\bibfield  {journal} {\bibinfo  {journal} {Nature
  Communications}\ }\textbf {\bibinfo {volume} {6}},\ \bibinfo {pages} {5886}
  (\bibinfo {year} {2015})}\BibitemShut {NoStop}%
\bibitem [{\citenamefont {Cavaill{\`e}s}\ \emph {et~al.}(2018)\citenamefont
  {Cavaill{\`e}s}, \citenamefont {Le~Jeannic}, \citenamefont {Raskop},
  \citenamefont {Guccione}, \citenamefont {Markham}, \citenamefont {Diamanti},
  \citenamefont {Shaw}, \citenamefont {Verma}, \citenamefont {Nam},\ and\
  \citenamefont {Laurat}}]{Cavailles:2018aa}%
  \BibitemOpen
  \bibfield  {author} {\bibinfo {author} {\bibfnamefont {A.}~\bibnamefont
  {Cavaill{\`e}s}}, \bibinfo {author} {\bibfnamefont {H.}~\bibnamefont
  {Le~Jeannic}}, \bibinfo {author} {\bibfnamefont {J.}~\bibnamefont {Raskop}},
  \bibinfo {author} {\bibfnamefont {G.}~\bibnamefont {Guccione}}, \bibinfo
  {author} {\bibfnamefont {D.}~\bibnamefont {Markham}}, \bibinfo {author}
  {\bibfnamefont {E.}~\bibnamefont {Diamanti}}, \bibinfo {author}
  {\bibfnamefont {M.~D.}\ \bibnamefont {Shaw}}, \bibinfo {author}
  {\bibfnamefont {V.~B.}\ \bibnamefont {Verma}}, \bibinfo {author}
  {\bibfnamefont {S.~W.}\ \bibnamefont {Nam}}, \ and\ \bibinfo {author}
  {\bibfnamefont {J.}~\bibnamefont {Laurat}},\ }\href {\doibase
  10.1103/PhysRevLett.121.170403} {\bibfield  {journal} {\bibinfo  {journal}
  {Physical Review Letters}\ }\textbf {\bibinfo {volume} {121}},\ \bibinfo
  {pages} {170403} (\bibinfo {year} {2018})}\BibitemShut {NoStop}%
\bibitem [{\citenamefont {Reid}\ and\ \citenamefont
  {Drummond}(1988)}]{PhysRevLett.60.2731}%
  \BibitemOpen
  \bibfield  {author} {\bibinfo {author} {\bibfnamefont {M.~D.}\ \bibnamefont
  {Reid}}\ and\ \bibinfo {author} {\bibfnamefont {P.~D.}\ \bibnamefont
  {Drummond}},\ }\href {\doibase 10.1103/PhysRevLett.60.2731} {\bibfield
  {journal} {\bibinfo  {journal} {Phys. Rev. Lett.}\ }\textbf {\bibinfo
  {volume} {60}},\ \bibinfo {pages} {2731} (\bibinfo {year}
  {1988})}\BibitemShut {NoStop}%
\bibitem [{\citenamefont {Frigerio}\ \emph {et~al.}(2020)\citenamefont
  {Frigerio}, \citenamefont {Olivares},\ and\ \citenamefont
  {Paris}}]{frigerio2020nonclassical}%
  \BibitemOpen
  \bibfield  {author} {\bibinfo {author} {\bibfnamefont {M.}~\bibnamefont
  {Frigerio}}, \bibinfo {author} {\bibfnamefont {S.}~\bibnamefont {Olivares}},
  \ and\ \bibinfo {author} {\bibfnamefont {M.~G.~A.}\ \bibnamefont {Paris}},\
  }\href@noop {} {\enquote {\bibinfo {title} {Nonclassical steering and the
  gaussian steering triangoloids},}\ } (\bibinfo {year} {2020}),\ \Eprint
  {http://arxiv.org/abs/2006.11912} {arXiv:2006.11912 [quant-ph]} \BibitemShut
  {NoStop}%
\bibitem [{\citenamefont {Walschaers}\ and\ \citenamefont
  {Treps}(2020)}]{PhysRevLett.124.150501}%
  \BibitemOpen
  \bibfield  {author} {\bibinfo {author} {\bibfnamefont {M.}~\bibnamefont
  {Walschaers}}\ and\ \bibinfo {author} {\bibfnamefont {N.}~\bibnamefont
  {Treps}},\ }\href {\doibase 10.1103/PhysRevLett.124.150501} {\bibfield
  {journal} {\bibinfo  {journal} {Phys. Rev. Lett.}\ }\textbf {\bibinfo
  {volume} {124}},\ \bibinfo {pages} {150501} (\bibinfo {year}
  {2020})}\BibitemShut {NoStop}%
\bibitem [{\citenamefont {Takagi}\ and\ \citenamefont
  {Zhuang}(2018)}]{Takagi:2018aa}%
  \BibitemOpen
  \bibfield  {author} {\bibinfo {author} {\bibfnamefont {R.}~\bibnamefont
  {Takagi}}\ and\ \bibinfo {author} {\bibfnamefont {Q.}~\bibnamefont
  {Zhuang}},\ }\href {\doibase 10.1103/PhysRevA.97.062337} {\bibfield
  {journal} {\bibinfo  {journal} {Physical Review A}\ }\textbf {\bibinfo
  {volume} {97}},\ \bibinfo {pages} {062337} (\bibinfo {year}
  {2018})}\BibitemShut {NoStop}%
\bibitem [{\citenamefont {Albarelli}\ \emph {et~al.}(2018)\citenamefont
  {Albarelli}, \citenamefont {Genoni}, \citenamefont {Paris},\ and\
  \citenamefont {Ferraro}}]{PhysRevA.98.052350}%
  \BibitemOpen
  \bibfield  {author} {\bibinfo {author} {\bibfnamefont {F.}~\bibnamefont
  {Albarelli}}, \bibinfo {author} {\bibfnamefont {M.~G.}\ \bibnamefont
  {Genoni}}, \bibinfo {author} {\bibfnamefont {M.~G.~A.}\ \bibnamefont
  {Paris}}, \ and\ \bibinfo {author} {\bibfnamefont {A.}~\bibnamefont
  {Ferraro}},\ }\href {\doibase 10.1103/PhysRevA.98.052350} {\bibfield
  {journal} {\bibinfo  {journal} {Phys. Rev. A}\ }\textbf {\bibinfo {volume}
  {98}},\ \bibinfo {pages} {052350} (\bibinfo {year} {2018})}\BibitemShut
  {NoStop}%
\bibitem [{\citenamefont {Kenfack}\ and\ \citenamefont
  {{\.Z}yczkowski}(2004)}]{Kenfack:2004aa}%
  \BibitemOpen
  \bibfield  {author} {\bibinfo {author} {\bibfnamefont {A.}~\bibnamefont
  {Kenfack}}\ and\ \bibinfo {author} {\bibfnamefont {K.}~\bibnamefont
  {{\.Z}yczkowski}},\ }\href {\doibase 10.1088/1464-4266/6/10/003} {\bibfield
  {journal} {\bibinfo  {journal} {Journal of Optics B: Quantum and
  Semiclassical Optics}\ }\textbf {\bibinfo {volume} {6}},\ \bibinfo {pages}
  {396} (\bibinfo {year} {2004})}\BibitemShut {NoStop}%
\bibitem [{\citenamefont {Opatrn\'y}\ \emph {et~al.}(2000)\citenamefont
  {Opatrn\'y}, \citenamefont {Kurizki},\ and\ \citenamefont
  {Welsch}}]{PhysRevA.61.032302}%
  \BibitemOpen
  \bibfield  {author} {\bibinfo {author} {\bibfnamefont {T.}~\bibnamefont
  {Opatrn\'y}}, \bibinfo {author} {\bibfnamefont {G.}~\bibnamefont {Kurizki}},
  \ and\ \bibinfo {author} {\bibfnamefont {D.-G.}\ \bibnamefont {Welsch}},\
  }\href {\doibase 10.1103/PhysRevA.61.032302} {\bibfield  {journal} {\bibinfo
  {journal} {Phys. Rev. A}\ }\textbf {\bibinfo {volume} {61}},\ \bibinfo
  {pages} {032302} (\bibinfo {year} {2000})}\BibitemShut {NoStop}%
\bibitem [{\citenamefont {Olivares}\ \emph {et~al.}(2003)\citenamefont
  {Olivares}, \citenamefont {Paris},\ and\ \citenamefont
  {Bonifacio}}]{PhysRevA.67.032314}%
  \BibitemOpen
  \bibfield  {author} {\bibinfo {author} {\bibfnamefont {S.}~\bibnamefont
  {Olivares}}, \bibinfo {author} {\bibfnamefont {M.~G.~A.}\ \bibnamefont
  {Paris}}, \ and\ \bibinfo {author} {\bibfnamefont {R.}~\bibnamefont
  {Bonifacio}},\ }\href {\doibase 10.1103/PhysRevA.67.032314} {\bibfield
  {journal} {\bibinfo  {journal} {Phys. Rev. A}\ }\textbf {\bibinfo {volume}
  {67}},\ \bibinfo {pages} {032314} (\bibinfo {year} {2003})}\BibitemShut
  {NoStop}%
\bibitem [{\citenamefont {Garc\'{\i}a-Patr\'on}\ \emph
  {et~al.}(2004)\citenamefont {Garc\'{\i}a-Patr\'on}, \citenamefont
  {Fiur\'a\ifmmode~\check{s}\else \v{s}\fi{}ek}, \citenamefont {Cerf},
  \citenamefont {Wenger}, \citenamefont {Tualle-Brouri},\ and\ \citenamefont
  {Grangier}}]{PhysRevLett.93.130409}%
  \BibitemOpen
  \bibfield  {author} {\bibinfo {author} {\bibfnamefont {R.}~\bibnamefont
  {Garc\'{\i}a-Patr\'on}}, \bibinfo {author} {\bibfnamefont {J.}~\bibnamefont
  {Fiur\'a\ifmmode~\check{s}\else \v{s}\fi{}ek}}, \bibinfo {author}
  {\bibfnamefont {N.~J.}\ \bibnamefont {Cerf}}, \bibinfo {author}
  {\bibfnamefont {J.}~\bibnamefont {Wenger}}, \bibinfo {author} {\bibfnamefont
  {R.}~\bibnamefont {Tualle-Brouri}}, \ and\ \bibinfo {author} {\bibfnamefont
  {P.}~\bibnamefont {Grangier}},\ }\href {\doibase
  10.1103/PhysRevLett.93.130409} {\bibfield  {journal} {\bibinfo  {journal}
  {Phys. Rev. Lett.}\ }\textbf {\bibinfo {volume} {93}},\ \bibinfo {pages}
  {130409} (\bibinfo {year} {2004})}\BibitemShut {NoStop}%
\bibitem [{\citenamefont {Kitagawa}\ \emph {et~al.}(2006)\citenamefont
  {Kitagawa}, \citenamefont {Takeoka}, \citenamefont {Sasaki},\ and\
  \citenamefont {Chefles}}]{PhysRevA.73.042310}%
  \BibitemOpen
  \bibfield  {author} {\bibinfo {author} {\bibfnamefont {A.}~\bibnamefont
  {Kitagawa}}, \bibinfo {author} {\bibfnamefont {M.}~\bibnamefont {Takeoka}},
  \bibinfo {author} {\bibfnamefont {M.}~\bibnamefont {Sasaki}}, \ and\ \bibinfo
  {author} {\bibfnamefont {A.}~\bibnamefont {Chefles}},\ }\href {\doibase
  10.1103/PhysRevA.73.042310} {\bibfield  {journal} {\bibinfo  {journal} {Phys.
  Rev. A}\ }\textbf {\bibinfo {volume} {73}},\ \bibinfo {pages} {042310}
  (\bibinfo {year} {2006})}\BibitemShut {NoStop}%
\bibitem [{\citenamefont {Yang}\ and\ \citenamefont
  {Li}(2009)}]{PhysRevA.80.022315}%
  \BibitemOpen
  \bibfield  {author} {\bibinfo {author} {\bibfnamefont {Y.}~\bibnamefont
  {Yang}}\ and\ \bibinfo {author} {\bibfnamefont {F.-L.}\ \bibnamefont {Li}},\
  }\href {\doibase 10.1103/PhysRevA.80.022315} {\bibfield  {journal} {\bibinfo
  {journal} {Phys. Rev. A}\ }\textbf {\bibinfo {volume} {80}},\ \bibinfo
  {pages} {022315} (\bibinfo {year} {2009})}\BibitemShut {NoStop}%
\bibitem [{\citenamefont {Navarrete-Benlloch}\ \emph
  {et~al.}(2012)\citenamefont {Navarrete-Benlloch}, \citenamefont
  {Garc\'{\i}a-Patr\'on}, \citenamefont {Shapiro},\ and\ \citenamefont
  {Cerf}}]{PhysRevA.86.012328}%
  \BibitemOpen
  \bibfield  {author} {\bibinfo {author} {\bibfnamefont {C.}~\bibnamefont
  {Navarrete-Benlloch}}, \bibinfo {author} {\bibfnamefont {R.}~\bibnamefont
  {Garc\'{\i}a-Patr\'on}}, \bibinfo {author} {\bibfnamefont {J.~H.}\
  \bibnamefont {Shapiro}}, \ and\ \bibinfo {author} {\bibfnamefont {N.~J.}\
  \bibnamefont {Cerf}},\ }\href {\doibase 10.1103/PhysRevA.86.012328}
  {\bibfield  {journal} {\bibinfo  {journal} {Phys. Rev. A}\ }\textbf {\bibinfo
  {volume} {86}},\ \bibinfo {pages} {012328} (\bibinfo {year}
  {2012})}\BibitemShut {NoStop}%
\bibitem [{\citenamefont {Das}\ \emph {et~al.}(2016)\citenamefont {Das},
  \citenamefont {Prabhu}, \citenamefont {Sen(De)},\ and\ \citenamefont
  {Sen}}]{PhysRevA.93.052313}%
  \BibitemOpen
  \bibfield  {author} {\bibinfo {author} {\bibfnamefont {T.}~\bibnamefont
  {Das}}, \bibinfo {author} {\bibfnamefont {R.}~\bibnamefont {Prabhu}},
  \bibinfo {author} {\bibfnamefont {A.}~\bibnamefont {Sen(De)}}, \ and\
  \bibinfo {author} {\bibfnamefont {U.}~\bibnamefont {Sen}},\ }\href {\doibase
  10.1103/PhysRevA.93.052313} {\bibfield  {journal} {\bibinfo  {journal} {Phys.
  Rev. A}\ }\textbf {\bibinfo {volume} {93}},\ \bibinfo {pages} {052313}
  (\bibinfo {year} {2016})}\BibitemShut {NoStop}%
\bibitem [{\citenamefont {Walschaers}\ \emph {et~al.}(2017)\citenamefont
  {Walschaers}, \citenamefont {Fabre}, \citenamefont {Parigi},\ and\
  \citenamefont {Treps}}]{walschaers_entanglement_2017}%
  \BibitemOpen
  \bibfield  {author} {\bibinfo {author} {\bibfnamefont {M.}~\bibnamefont
  {Walschaers}}, \bibinfo {author} {\bibfnamefont {C.}~\bibnamefont {Fabre}},
  \bibinfo {author} {\bibfnamefont {V.}~\bibnamefont {Parigi}}, \ and\ \bibinfo
  {author} {\bibfnamefont {N.}~\bibnamefont {Treps}},\ }\href {\doibase
  10.1103/PhysRevLett.119.183601} {\bibfield  {journal} {\bibinfo  {journal}
  {Phys. Rev. Lett.}\ }\textbf {\bibinfo {volume} {119}},\ \bibinfo {pages}
  {183601} (\bibinfo {year} {2017})}\BibitemShut {NoStop}%
\bibitem [{\citenamefont {Ourjoumtsev}\ \emph {et~al.}(2007)\citenamefont
  {Ourjoumtsev}, \citenamefont {Dantan}, \citenamefont {Tualle-Brouri},\ and\
  \citenamefont {Grangier}}]{ourjoumtsev_increasing_2007}%
  \BibitemOpen
  \bibfield  {author} {\bibinfo {author} {\bibfnamefont {A.}~\bibnamefont
  {Ourjoumtsev}}, \bibinfo {author} {\bibfnamefont {A.}~\bibnamefont {Dantan}},
  \bibinfo {author} {\bibfnamefont {R.}~\bibnamefont {Tualle-Brouri}}, \ and\
  \bibinfo {author} {\bibfnamefont {P.}~\bibnamefont {Grangier}},\ }\href
  {\doibase 10.1103/PhysRevLett.98.030502} {\bibfield  {journal} {\bibinfo
  {journal} {Phys. Rev. Lett.}\ }\textbf {\bibinfo {volume} {98}},\ \bibinfo
  {pages} {030502} (\bibinfo {year} {2007})}\BibitemShut {NoStop}%
\bibitem [{\citenamefont {Takahashi}\ \emph {et~al.}(2010)\citenamefont
  {Takahashi}, \citenamefont {Neergaard-Nielsen}, \citenamefont {Takeuchi},
  \citenamefont {Takeoka}, \citenamefont {Hayasaka}, \citenamefont {Furusawa},\
  and\ \citenamefont {Sasaki}}]{takahashi_entanglement_2010}%
  \BibitemOpen
  \bibfield  {author} {\bibinfo {author} {\bibfnamefont {H.}~\bibnamefont
  {Takahashi}}, \bibinfo {author} {\bibfnamefont {J.~S.}\ \bibnamefont
  {Neergaard-Nielsen}}, \bibinfo {author} {\bibfnamefont {M.}~\bibnamefont
  {Takeuchi}}, \bibinfo {author} {\bibfnamefont {M.}~\bibnamefont {Takeoka}},
  \bibinfo {author} {\bibfnamefont {K.}~\bibnamefont {Hayasaka}}, \bibinfo
  {author} {\bibfnamefont {A.}~\bibnamefont {Furusawa}}, \ and\ \bibinfo
  {author} {\bibfnamefont {M.}~\bibnamefont {Sasaki}},\ }\href {\doibase
  10.1038/nphoton.2010.1} {\bibfield  {journal} {\bibinfo  {journal} {Nat
  Photon}\ }\textbf {\bibinfo {volume} {4}},\ \bibinfo {pages} {178} (\bibinfo
  {year} {2010})}\BibitemShut {NoStop}%
\bibitem [{\citenamefont {Morin}\ \emph {et~al.}(2014)\citenamefont {Morin},
  \citenamefont {Huang}, \citenamefont {Liu}, \citenamefont {Le~Jeannic},
  \citenamefont {Fabre},\ and\ \citenamefont {Laurat}}]{Morin:2014aa}%
  \BibitemOpen
  \bibfield  {author} {\bibinfo {author} {\bibfnamefont {O.}~\bibnamefont
  {Morin}}, \bibinfo {author} {\bibfnamefont {K.}~\bibnamefont {Huang}},
  \bibinfo {author} {\bibfnamefont {J.}~\bibnamefont {Liu}}, \bibinfo {author}
  {\bibfnamefont {H.}~\bibnamefont {Le~Jeannic}}, \bibinfo {author}
  {\bibfnamefont {C.}~\bibnamefont {Fabre}}, \ and\ \bibinfo {author}
  {\bibfnamefont {J.}~\bibnamefont {Laurat}},\ }\href {\doibase
  10.1038/nphoton.2014.137} {\bibfield  {journal} {\bibinfo  {journal} {Nature
  Photonics}\ }\textbf {\bibinfo {volume} {8}},\ \bibinfo {pages} {570}
  (\bibinfo {year} {2014})}\BibitemShut {NoStop}%
\bibitem [{\citenamefont {Gagatsos}\ and\ \citenamefont
  {Guha}(2019)}]{PhysRevA.99.053816}%
  \BibitemOpen
  \bibfield  {author} {\bibinfo {author} {\bibfnamefont {C.~N.}\ \bibnamefont
  {Gagatsos}}\ and\ \bibinfo {author} {\bibfnamefont {S.}~\bibnamefont
  {Guha}},\ }\href {\doibase 10.1103/PhysRevA.99.053816} {\bibfield  {journal}
  {\bibinfo  {journal} {Phys. Rev. A}\ }\textbf {\bibinfo {volume} {99}},\
  \bibinfo {pages} {053816} (\bibinfo {year} {2019})}\BibitemShut {NoStop}%
\bibitem [{\citenamefont {Fedorov}\ \emph {et~al.}(2015)\citenamefont
  {Fedorov}, \citenamefont {Ulanov}, \citenamefont {Kurochkin},\ and\
  \citenamefont {Lvovsky}}]{Fedorov:15}%
  \BibitemOpen
  \bibfield  {author} {\bibinfo {author} {\bibfnamefont {I.~A.}\ \bibnamefont
  {Fedorov}}, \bibinfo {author} {\bibfnamefont {A.~E.}\ \bibnamefont {Ulanov}},
  \bibinfo {author} {\bibfnamefont {Y.~V.}\ \bibnamefont {Kurochkin}}, \ and\
  \bibinfo {author} {\bibfnamefont {A.~I.}\ \bibnamefont {Lvovsky}},\ }\href
  {\doibase 10.1364/OPTICA.2.000112} {\bibfield  {journal} {\bibinfo  {journal}
  {Optica}\ }\textbf {\bibinfo {volume} {2}},\ \bibinfo {pages} {112} (\bibinfo
  {year} {2015})}\BibitemShut {NoStop}%
\bibitem [{\citenamefont {Katamadze}\ \emph {et~al.}(2018)\citenamefont
  {Katamadze}, \citenamefont {Avosopiants}, \citenamefont {Bogdanov},\ and\
  \citenamefont {Kulik}}]{Katamadze:18}%
  \BibitemOpen
  \bibfield  {author} {\bibinfo {author} {\bibfnamefont {K.~G.}\ \bibnamefont
  {Katamadze}}, \bibinfo {author} {\bibfnamefont {G.~V.}\ \bibnamefont
  {Avosopiants}}, \bibinfo {author} {\bibfnamefont {Y.~I.}\ \bibnamefont
  {Bogdanov}}, \ and\ \bibinfo {author} {\bibfnamefont {S.~P.}\ \bibnamefont
  {Kulik}},\ }\href {\doibase 10.1364/OPTICA.5.000723} {\bibfield  {journal}
  {\bibinfo  {journal} {Optica}\ }\textbf {\bibinfo {volume} {5}},\ \bibinfo
  {pages} {723} (\bibinfo {year} {2018})}\BibitemShut {NoStop}%
\bibitem [{\citenamefont {Wick}(1950)}]{PhysRev.80.268}%
  \BibitemOpen
  \bibfield  {author} {\bibinfo {author} {\bibfnamefont {G.~C.}\ \bibnamefont
  {Wick}},\ }\href {\doibase 10.1103/PhysRev.80.268} {\bibfield  {journal}
  {\bibinfo  {journal} {Phys. Rev.}\ }\textbf {\bibinfo {volume} {80}},\
  \bibinfo {pages} {268} (\bibinfo {year} {1950})}\BibitemShut {NoStop}%
\bibitem [{\citenamefont {Lov{\'a}sz}\ and\ \citenamefont
  {Plummer}(2009)}]{Matching-Book}%
  \BibitemOpen
  \bibfield  {author} {\bibinfo {author} {\bibfnamefont {L.}~\bibnamefont
  {Lov{\'a}sz}}\ and\ \bibinfo {author} {\bibfnamefont {M.~D.}\ \bibnamefont
  {Plummer}},\ }\href@noop {} {\emph {\bibinfo {title} {Matching Theory}}}\
  (\bibinfo  {publisher} {AMS Chelsea Publishing},\ \bibinfo {address}
  {Providence, RI},\ \bibinfo {year} {2009})\BibitemShut {NoStop}%
\bibitem [{\citenamefont {Hamilton}\ \emph {et~al.}(2017)\citenamefont
  {Hamilton}, \citenamefont {Kruse}, \citenamefont {Sansoni}, \citenamefont
  {Barkhofen}, \citenamefont {Silberhorn},\ and\ \citenamefont
  {Jex}}]{PhysRevLett.119.170501}%
  \BibitemOpen
  \bibfield  {author} {\bibinfo {author} {\bibfnamefont {C.~S.}\ \bibnamefont
  {Hamilton}}, \bibinfo {author} {\bibfnamefont {R.}~\bibnamefont {Kruse}},
  \bibinfo {author} {\bibfnamefont {L.}~\bibnamefont {Sansoni}}, \bibinfo
  {author} {\bibfnamefont {S.}~\bibnamefont {Barkhofen}}, \bibinfo {author}
  {\bibfnamefont {C.}~\bibnamefont {Silberhorn}}, \ and\ \bibinfo {author}
  {\bibfnamefont {I.}~\bibnamefont {Jex}},\ }\href {\doibase
  10.1103/PhysRevLett.119.170501} {\bibfield  {journal} {\bibinfo  {journal}
  {Phys. Rev. Lett.}\ }\textbf {\bibinfo {volume} {119}},\ \bibinfo {pages}
  {170501} (\bibinfo {year} {2017})}\BibitemShut {NoStop}%
\bibitem [{\citenamefont {Br\'adler}\ \emph {et~al.}(2018)\citenamefont
  {Br\'adler}, \citenamefont {Dallaire-Demers}, \citenamefont {Rebentrost},
  \citenamefont {Su},\ and\ \citenamefont {Weedbrook}}]{PhysRevA.98.032310}%
  \BibitemOpen
  \bibfield  {author} {\bibinfo {author} {\bibfnamefont {K.}~\bibnamefont
  {Br\'adler}}, \bibinfo {author} {\bibfnamefont {P.-L.}\ \bibnamefont
  {Dallaire-Demers}}, \bibinfo {author} {\bibfnamefont {P.}~\bibnamefont
  {Rebentrost}}, \bibinfo {author} {\bibfnamefont {D.}~\bibnamefont {Su}}, \
  and\ \bibinfo {author} {\bibfnamefont {C.}~\bibnamefont {Weedbrook}},\ }\href
  {\doibase 10.1103/PhysRevA.98.032310} {\bibfield  {journal} {\bibinfo
  {journal} {Phys. Rev. A}\ }\textbf {\bibinfo {volume} {98}},\ \bibinfo
  {pages} {032310} (\bibinfo {year} {2018})}\BibitemShut {NoStop}%
\bibitem [{\citenamefont {Chabaud}\ \emph {et~al.}(2017)\citenamefont
  {Chabaud}, \citenamefont {Douce}, \citenamefont {Markham}, \citenamefont {van
  Loock}, \citenamefont {Kashefi},\ and\ \citenamefont
  {Ferrini}}]{PhysRevA.96.062307}%
  \BibitemOpen
  \bibfield  {author} {\bibinfo {author} {\bibfnamefont {U.}~\bibnamefont
  {Chabaud}}, \bibinfo {author} {\bibfnamefont {T.}~\bibnamefont {Douce}},
  \bibinfo {author} {\bibfnamefont {D.}~\bibnamefont {Markham}}, \bibinfo
  {author} {\bibfnamefont {P.}~\bibnamefont {van Loock}}, \bibinfo {author}
  {\bibfnamefont {E.}~\bibnamefont {Kashefi}}, \ and\ \bibinfo {author}
  {\bibfnamefont {G.}~\bibnamefont {Ferrini}},\ }\href {\doibase
  10.1103/PhysRevA.96.062307} {\bibfield  {journal} {\bibinfo  {journal} {Phys.
  Rev. A}\ }\textbf {\bibinfo {volume} {96}},\ \bibinfo {pages} {062307}
  (\bibinfo {year} {2017})}\BibitemShut {NoStop}%
\end{thebibliography}%

\end{document}